\documentclass[12pt,onecolumn,twoside]{IEEEtran}

\usepackage{calc}

\usepackage{multirow}
\usepackage{cite}
\usepackage{graphicx,subfigure}
\usepackage{psfrag}
\usepackage{amsmath,amssymb,amsthm}
\usepackage{color}
\usepackage{tikz} 
\usepackage{bm}
\usepackage{bbm}
\usepackage{verbatim}
\usepackage[T1]{fontenc}

\usepackage{cases}
\usepackage[noend]{algpseudocode}

\usepackage{xcolor}
\usepackage[most]{tcolorbox}
\usepackage{bm}

\usepackage{tabu}
\usepackage{mathtools}
\usepackage{xifthen}
\usepackage{booktabs}
\usepackage{algorithm}
\usepackage{algpseudocode}
\usepackage{arydshln}

\makeatletter
\def\BState{\State\hskip-\ALG@thistlm}
\makeatother

\DeclareMathOperator*{\defeq}{\triangleq}

\newtheorem{theorem}{Theorem}

\newtheorem{lemma}{Lemma}
\newtheorem{definition}{Definition}

\newcommand{\bit}{\begin{itemize}}
\newcommand{\eit}{\end{itemize}}

\newcommand{\bc}{\begin{center}}
\newcommand{\ec}{\end{center}}

\newcommand{\ba}{\begin{array}}
\newcommand{\ea}{\end{array}}

\newcommand{\beq}{\begin{equation}}
\newcommand{\eeq}{\end{equation}}

\newcommand{\beqn}{\begin{equation*}}
\newcommand{\eeqn}{\end{equation*}}

\newcommand{\bean}{\begin{eqnarray*}}
\newcommand{\eean}{\end{eqnarray*}}
\newcommand{\bea}{\begin{eqnarray}}
\newcommand{\eea}{\end{eqnarray}}

\def\cv{\boldsymbol{c}}

\def\vv{\boldsymbol{v}}
\def\wv{\boldsymbol{w}}

\newcommand{\Ac}{{\mathcal A}}
\newcommand{\Bc}{{\mathcal B}}

\newcommand{\Fc}{{\mathcal F}}
\newcommand{\Gc}{{\mathcal G}}

\newcommand{\Ic}{{\mathcal I}}

\newcommand{\Rc}{{\mathcal R}}

\newcommand{\Vc}{{\mathcal V}}

\algnewcommand{\IfThenElse}[3]{  \State \algorithmicif\ #1\ \algorithmicthen\ #2\ \algorithmicelse\ #3}

\newcommand\bl[1]{{\color{blue}#1}}

\newcommand{\Me}{\wv}

\newcommand{\rtuple}{)} 
\newcommand{\ltuple}{(}

\newcommand{\inset}{\in}

\newcommand{\OR}{\lor}

\newcommand{\eqlog}{=}

\newcommand{\MVBAInputMsg}{\wv}

\newcommand{\EncodedSymbol}{y}

\newcommand{\VOTE}{\text{``}\mathrm{VOTE}\text{''}}

\newcommand{\CONFIRM}{\text{``}\mathrm{CONFIRM}\text{''}}

\newcommand{\electionround}{r} 
   
\newcommand{\Election}{\mathrm{Election}}               
\newcommand{\electionoutput}{l}

\newcommand{\ABBA}{\mathrm{ABBA}}

\newcommand{\ACD}{\mathrm{ACID}}

\newcommand{\thisnodeindex}{i}                       
\newcommand{\READY}{\text{``}\mathrm{READY}\text{''}}                           
\newcommand{\send}{\textbf{send}}                   
                    
\newcommand{\ReadyRecord}{\Rc_\mathrm{ready}}                                                                                                           
                        
\newcommand{\IDMVBA}{\mathrm{ID}}    
  
\newcommand{\IDMVBAtilde}{\mathrm{id}}

\newcommand{\FinishRecord}{\Fc_\mathrm{finish}}                                                                                                           
                       
\newcommand{\ELECTION}{\text{``}\mathrm{ELECTION}\text{''}}                        
\newcommand{\FINISH}{\text{``}\mathrm{FINISH}\text{''}}                        
                        
\newcommand{\ABBAoutput}{a}   
  
\newcommand{\ABBBA}{\mathrm{ABBBA}}       
\newcommand{\ABAoutput}{b}    
                    
\newcommand{\Output}{\textbf{output}}                      
\newcommand{\terminate}{\textbf{terminate}}

\newcommand{\abbainput}{a}     
\newcommand{\abbainputA}{a_1}    
\newcommand{\abbainputB}{a_2}      
\newcommand{\ABBAVALUE}{\text{``}\mathrm{ABBA}\text{''}}                        
                        
\newcommand{\ABBACountA}{\mathrm{cnt}_1}    
\newcommand{\ABBACountB}{\mathrm{cnt}_2}      
\newcommand{\ABBACountC}{\mathrm{cnt}_3}

\newcommand{\wait}{\textbf{wait}}

\newcommand{\EC}{\mathrm{EC}}  
\newcommand{\ECEnc}{\mathrm{ECEnc}}   
\newcommand{\ECDec}{\mathrm{ECDec}}

\newcommand{\Alphabet}{\Bc}

\newcommand{\defaultvalue}{\bot}

\newcommand{\OciorABA}{\mathrm{OciorABA}}         
\newcommand{\OciorABAstar}{\mathrm{OciorABA}^\star}

\newcommand{\Pass}{\textbf{pass}}

\newcommand{\COOL}{\mathrm{COOL}}

\newcommand{\networksizen}{n}                                           
                                            
\newcommand{\networkfaultsizet}{t}

\newcommand{\BA}{\mathrm{BA}}

\newcommand{\OciorRBC}{\text{OciorRBC}}

\newcommand{\ABA}{\mathrm{ABA}}

\newcommand{\IDABA}{\mathrm{ID}}                                                                        
\newcommand{\RBC}{\mathrm{RBC}}                                                                          
                                                                          
\newcommand{\ABAOneSet}{\Ac_{\mathrm{ones}}}                                                                          
\newcommand{\ABAOneSetK}{\Bc_{\mathrm{ones}}}

\newcommand{\deliver}{\textbf{deliver}}

\newcommand{\alphabetsize}{q}

\newcommand{\Node}{P}

\newcommand{\APVA}{\mathrm{APVA}}                  
\newcommand{\missing}{\phi}                       
\newcommand{\Mc}{{\mathcal M}}    
                      
\newcommand{\binaryvalue}{b}                      
\newcommand{\ConfirmRecord}{\cv}                                                                                                           
\newcommand{\ConfirmCount}{\mathrm{cnt}}                       
\newcommand{\APBVAinput}{\vv}          
\newcommand{\APBVAoutput}{\hat{\vv}}                                  
          
\newcommand{\IDABANew}{\mathrm{ID}^\star}                                                                         
 \newcommand{\RBCReadyindicator}{\Rc_\mathrm{ready}^{\star}}      
\newcommand{\RBCFinishindicator}{\Fc_\mathrm{finish}^{\star}}      
\newcommand{\ABBAoutputNew}{a^\star}    
\newcommand{\ABAoutputNew}{b^\star}    
    
\newcommand{\ACIDd}{\mathrm{ACID}^{\star}}      
\newcommand{\NonmissingElementSet}{\Mc}                            
\newcommand{\symbolsize}{c}                            
\newcommand{\CommonSet}{\Gc}

\algdef{SE}[SUBALG]{Indent}{EndIndent}{}{\algorithmicend\ }%
\algtext*{Indent}
\algtext*{EndIndent}

\begin{document}
\sloppy
\title{OciorABA: Improved Error-Free Asynchronous Byzantine Agreement via  Partial Vector Agreement}

\author{Jinyuan Chen 
%\thanks{Jinyuan Chen is with Ocior Inc.  (email:  jinyuan@ocior.com ).} 
}

\maketitle
\pagestyle{headings}

%%%%%%%%%%%%%%%%%%%%%%%%%%%%%%%%%%%%%%%%%%%%%
\begin{abstract}
In this work, we propose an error-free, information-theoretically secure  multi-valued  asynchronous Byzantine agreement ($\ABA$) protocol, called $\OciorABA$. This protocol achieves $\ABA$ consensus on an $\ell$-bit message  with an expected communication complexity of   $O(n\ell + n^3 \log \alphabetsize )$   bits and an expected round complexity of   $O(1)$ rounds, under the optimal resilience condition $n \geq 3t + 1$ in an $n$-node network, where up to $t$ nodes may be dishonest. Here, $\alphabetsize$ denotes the alphabet size of the error correction code used in the protocol.  
In our protocol design, we introduce a new primitive: asynchronous partial vector agreement ($\APVA$). In  $\APVA$,  the distributed nodes input their vectors  and aim to  output a common vector, where  some of the elements of those vectors may be missing or unknown.  
We  propose an $\APVA$ protocol with  an expected communication complexity of $O( n^3 \log \alphabetsize )$ bits and an expected round complexity of  $O(1)$ rounds.  
This  $\APVA$ protocol serves as a key building block for our $\OciorABA$ protocol. 
\end{abstract}

%\begin{IEEEkeywords}
% Multi-valued   asynchronous Byzantine agreement,      information-theoretic security,   signature-free protocol, error-free protocol, error correction codes.  
%\end{IEEEkeywords}

\section{Introduction}

Byzantine agreement ($\BA$), first introduced by Pease, Shostak and Lamport in 1980,  is a distributed consensus problem where $n$ nodes aim to agree  on a message, even if  up to $t$ nodes are dishonest \cite{PSL:80}.   
$\BA$ and its variants serve as fundamental components of distributed systems and cryptography. 
This work focuses on the design of \emph{asynchronous} Byzantine agreement ($\ABA$) protocols, with an emphasis on error-free, information-theoretically secure (IT secure)  $\ABA$. An $\ABA$ protocol is considered \emph{IT secure}  if it meets all required properties  without relying on cryptographic assumptions such as signatures or hashing, except for the common coin assumption. Under the common coin assumption, an $\ABA$ protocol is classified as \emph{error-free} if it guarantees all required properties in \emph{all} executions.

 In the setting of error-free synchronous $\BA$, significant progress has been made  in understating its fundamental limits \cite{LV:11,GP:20,LDK:20,NRSVX:20,Chen:2020arxiv, ChenDISC:21,ChenOciorCOOL:24}. 
 Specifically,   the $\COOL$  protocol proposed by Chen  \cite{Chen:2020arxiv, ChenDISC:21,ChenOciorCOOL:24},  is a deterministic, error-free, IT secure  synchronous $\BA$   protocol that  achieves $\BA$ consensus on an $\ell$-bit message with a communication complexity of   $O(n \ell  + n t \log \alphabetsize)$   bits and  a round complexity of  $O(t)$ rounds,  under the optimal resilience condition $n \geq 3t + 1$.  
 Here, $\alphabetsize$ denotes the alphabet size of the error correction code used in the protocol.  
 When error correction codes with a constant alphabet size (e.g., Expander Code \cite{SS:96}) are used, $\alphabetsize$ becomes a constant, making  $\COOL$   optimal. 
 
  However,   progress in understanding the fundamental limits of error-free $\ABA$ has been more limited.   In this setting,  the protocol proposed by Patra \cite{ Patra:11},  achieves  $\ABA$ consensus with an expected  communication complexity of   $O(n\ell + n^5 \log n )$    bits and  an expected  round complexity of  $O(1)$ rounds, given $n \geq 3t + 1$. 
   Nayak  et al.       \cite{NRSVX:20}  later  improved  the communication complexity  to     $O(n\ell + n^4 \log n )$ bits. 
   Li and Chen \cite{LCabaISIT:21} further reduced the   communication complexity to  $O(\max\{n\ell, nt \log \alphabetsize\})$   bits, but under a suboptimal resilience condition $n \geq 5t + 1$.

 In this work, at first we propose an error-free,  IT secure  $\ABA$  protocol, called $\OciorABAstar$. 
 $\OciorABAstar$ does not rely on any   cryptographic assumptions, such as  signatures or hashing,  except for  the common coin assumption. 
This protocol achieves  $\ABA$ consensus on an $\ell$-bit message with an expected communication complexity of $O(n\ell + n^3 \log \alphabetsize )$  bits  and an expected  round complexity of $O(\log n)$ rounds,  under the optimal resilience  condition $n \geq 3t + 1$.  
Then,   we propose an improved error-free,  IT secure  $\ABA$  protocol, called $\OciorABA$. 
This $\OciorABA$ protocol achieves  $\ABA$ consensus   with an expected communication complexity of $O(n\ell + n^3 \log \alphabetsize )$  bits  and an expected  round complexity of $O(1)$ rounds,  under the optimal resilience  condition $n \geq 3t + 1$.   

In our protocol design for $\OciorABA$, we introduce a new primitive: asynchronous partial vector agreement ($\APVA$). 
In  $\APVA$,  the distributed nodes input their vectors  and aim to  output a common vector, where  some of the elements of those vectors may be missing or unknown.  
We  propose an $\APVA$ protocol with  an expected communication complexity of   $O( n^3 \log \alphabetsize )$  bits and an expected round complexity of     $O(1)$ rounds.  
This  $\APVA$ protocol serves as a key building block for our $\OciorABA$ protocol.

The proposed $\OciorABAstar$,  $\OciorABA$, and $\APVA$  protocols are described in Algorithm~\ref{algm:OciorABAstar}, Algorithm~\ref{algm:OciorABA},   and   Algorithm~\ref{algm:APVA}, respectively.   
 Table~\ref{tb:ABA}   provides a comparison between the proposed $\OciorABAstar$ and  $\OciorABA$ protocols, as well as some other $\ABA$ protocols. 
Some definitions and primitives are provided in the following subsection.

{\renewcommand{\arraystretch}{1.3}
\begin{table}
\footnotesize  
\begin{center}
\caption{Comparison between the proposed protocols and some other multi-valued    $\ABA$  protocols.     Here $\alphabetsize$ denotes the alphabet size of  the    error correction code used in the   protocols. 
When  error correction codes with a constant alphabet size are used, $\alphabetsize$ becomes a constant.  Here $\ell$ denotes the size of the messages to agree on.  
} \label{tb:ABA}
\begin{tabular}{||c||c|c|c|c|c|}
\hline
Protocols & Resilience &   Communication     & Rounds   &       Error-Free     &   Cryptographic     Assumption    \\ 
  &  &  (Expected Total Bits)  &   (Expected)  &               &            (Expect for     Common Coin)       \\ 
\hline
Patra \cite{ Patra:11}  &  $t<\frac{n}{3}$  &   $O(n\ell + n^5 \log n )$      &  $O(1)$    & Yes   &  Non   \\
\hline
Nayak  et al.       \cite{NRSVX:20}  &  $t<\frac{n}{3}$   &   $O(n\ell + n^4 \log n )$      &  $O(1)$     &  Yes  &  Non    \\
\hline
Li-Chen \cite{LCabaISIT:21}  &  $t<\frac{n}{5}$   &    $O(\max\{n\ell, nt \log \alphabetsize\})$        &   $O(1)$     &  Yes   &  Non        \\ 
\hline
    \bl{Proposed $\OciorABAstar$}   &   {\color{blue} $t<\frac{n}{3}$}  &   $\bl{O(n\ell + n^3 \log \alphabetsize )}$   &  {\color{blue} $O(\log n)$}     &   {\color{blue} Yes}     &    {\color{blue} Non}           \\
\hline
  \bl{Proposed $\OciorABA$}   &   {\color{blue} $t<\frac{n}{3}$}  &    $\bl{O(n\ell + n^3 \log \alphabetsize )}$   &   {\color{blue} $O(1)$}     &   {\color{blue} Yes}     &    {\color{blue} Non}            \\
\hline
\end{tabular}
\end{center}
\end{table}
}

\subsection{Primitives}

  We consider an \emph{asynchronous} network, where messages can be delayed arbitrarily by an adversary but are guaranteed to eventually arrive at their destinations.
The adversary is assumed to be adaptive, capable of corrupting any node at any time, with the constraint that at most $t$ nodes can be controlled in total.

For a vector  $\Me$ of size $n$, we use  $\Me[j] \in \Vc \cup \{\missing\}$ to denote the $j$th element of $\Me$, where  $\Vc$ is a non-empty alphabet $\Vc$ and $\missing\notin \Vc$.  The notion   $\Me[j]=\missing$ indicates  that the $j$th element of $\Me$ is missing or unknown.  Here, we focus on binary vectors in the asynchronous partial vector agreement  problem, where $\Vc=\{0,1\}$. 
We use $\NonmissingElementSet(\Me)$ to denote the set of indices of all non-missing  elements of  a vector $\Me$,  i.e., $\NonmissingElementSet(\Me):=\{j:  \Me[j] \neq \missing, j \in[1:n]\}$.   
 We use $\Fc$ to denote the set of indices of all dishonest nodes.

\begin{definition} [{\bf  Asynchronous partial vector agreement ($\APVA$)}]     \label{def:APVA}  
We introduce a new primitive called  $\APVA$.  
In the $\APVA$ problem,  each Node~$i$, $i \in [1:n]$,  inputs  a vector  $\Me_i$ of size $n$, where  some of the elements may be missing or unknown, i.e., $\Me_i[j]=\missing$ for some $j\in [1:n]$. 
The number of missing elements of  $\Me_i$ may decrease over time.  
The distributed nodes aim to  output a common vector  $\Me$, where some of the elements may be missing. 
The $\APVA$ protocol guarantees  the following properties: 
 \begin{itemize}
 \item  {\bf Consistency:} If any honest node outputs a vector $\wv$, then every honest node eventually outputs $\wv$.
\item  {\bf Validity:}    If an honest node outputs $\Me$, then for any non-missing element of $\Me$, i.e., $\Me[j] \neq \missing$ for some $j\in [1:n]$,   at least one honest Node~$i$,   for $i\in [1:n]\setminus \Fc$, must have input $\Me_i[j] =\Me[j]\neq \missing$.    Furthermore,   the number of  non-missing  elements of  the output $\Me$ is greater than or equal to  $n-t$, i.e.,   $|\NonmissingElementSet(\Me)|\geq n-t$,  where $\NonmissingElementSet(\Me):=\{j:  \Me[j] \neq \missing, j \in[1:n]\}$.   
\item  {\bf Termination:} If all honest nodes have input non-missing values in their input vectors for at least $n-t$  positions in common, then eventually every node will output a vector and terminate.
\end{itemize} 
\end{definition}

\begin{definition}  [{\bf Byzantine agreement ($\BA$)}]
The  $\BA$ protocol guarantees  the following properties: 
\begin{itemize}
\item  {\bf Termination:} If all  honest nodes receive their inputs, then every honest node  eventually outputs a value and terminates. 
\item  {\bf Consistency:} If any honest node outputs a value $\wv$, then every honest node eventually outputs $\wv$.
\item  {\bf Validity:}     If all honest nodes input the same  value $\wv$, then every honest node eventually outputs $\wv$.  
\end{itemize} 
\end{definition}

\begin{definition} [{\bf Reliable broadcast ($\RBC$)}]
 In a reliable broadcast protocol, a leader  inputs a value  and broadcasts it to distributed nodes,   satisfying the following conditions:
\begin{itemize}
\item  {\bf Consistency:} If any two honest nodes output $\wv'$ and $\wv''$, respectively, then  $\wv'=\wv''$.
\item   {\bf Validity:} If the leader is  honest and inputs a value $\wv$, then every honest node eventually outputs $\wv$. 
\item  {\bf Totality:}  If one  honest node outputs a value, then every honest node  eventually outputs a value.          
\end{itemize} 
\end{definition}

\noindent  {\bf Asynchronous  complete  information dispersal} ($\ACD$, \cite{ChenOciorMVBA:24}).  The objective of an $\ACD$ protocol is to disperse information across a network. Once  a leader completes the dispersal of its proposed message, every honest node is guaranteed to correctly reconstruct the delivered message from the distributed nodes using a data retrieval scheme.

\begin{definition} [$\ACD$ instance, \cite{ChenOciorMVBA:24}]
In an $\ACD[\ltuple \IDMVBA, i \rtuple]$ protocol, $\Node_i$ disperses a  message   over   distributed nodes, for $i\in [1:n]$. 
An $\ACD[\ltuple \IDMVBA, i \rtuple]$ protocol is complemented by a data retrieval protocol,   ensuring  the following properties:    
\begin{itemize}
\item  {\bf Completeness:}  If $\Node_i$ is honest, then $\Node_i$ eventually completes the dispersal $\ltuple \IDMVBA, i \rtuple$. 
\item  {\bf Availability:} If $\Node_i$ completes the dispersal for $\ltuple \IDMVBA, i \rtuple$,  and all honest nodes start the data retrieval  protocol  for $\ltuple \IDMVBA, i \rtuple$, then each node eventually reconstructs some message.
\item   {\bf Consistency:} If two honest nodes reconstruct messages $\MVBAInputMsg'$ and $\MVBAInputMsg''$ respectively for $\ltuple \IDMVBA, i \rtuple$, then  $\MVBAInputMsg'=\MVBAInputMsg''$. 
\item   {\bf Validity:} If an honest $\Node_i$ has proposed a message $\MVBAInputMsg$ for $\ltuple \IDMVBA, i \rtuple$ and an honest node reconstructs a message $\MVBAInputMsg'$  for $\ltuple \IDMVBA, i \rtuple$, then  $\MVBAInputMsg'=\MVBAInputMsg$. 
\end{itemize} 
\end{definition} 

\begin{definition} [Parallel $\ACD$ instances, \cite{ChenOciorMVBA:24}]
An $\ACD[ \IDMVBA]$ protocol  involves running  $n$ parallel  $\ACD$ instances, $\{\ACD[\ltuple \IDMVBA, i \rtuple]\}_{i=1}^{n}$, satisfying the following properties:  
\begin{itemize}
\item  {\bf Termination:} Every honest node eventually terminates.     
\item  {\bf Integrity:} If one honest node terminates, then there exists a set $\Ic^{\star}$ such that the following conditions hold: 1) $\Ic^{\star}\subseteq [1:n]\setminus \Fc$; 2)  $|\Ic^{\star}| \geq n-2t$; and  3) for any $i\in \Ic^{\star}$,  $\Node_i$ has completed the dispersal $\ACD[\ltuple \IDMVBA, i \rtuple]$.      
\end{itemize}  
\end{definition}

\begin{definition} [{\bf Asynchronous  biased  binary Byzantine agreement ($\ABBBA$)},  \cite{ChenOciorMVBA:24}]    
In an $\ABBBA$ protocol, each  honest node provides a pair of binary inputs $(\abbainputA, \abbainputB)$, where  $\abbainputA, \abbainputB \in \{0,1\}$. The  honest nodes aim to reach a consensus  on a common value $\abbainput \in \{0,1\}$, satisfying the following properties:  
\begin{itemize}
\item   {\bf Conditional termination:} Under an input condition---i.e.,  if one honest node inputs its second number as $\abbainputB =1$ then at least $t+1$ honest nodes  input  their first numbers as $\abbainputA =1$---every honest node eventually outputs a value and terminates.     
\item   {\bf Biased validity:} If at least $t+1$ honest nodes input the second number as $\abbainputB=1$, then any honest node that terminates outputs $1$.       
\item   {\bf Biased integrity:} If any honest node outputs $1$, then at least one honest node inputs $\abbainputA=1$ or $\abbainputB=1$.             
\end{itemize} 
\end{definition}

 \begin{definition} [{\bf Common coin}]
We assume the existence of a common coin protocol $\electionoutput \gets \Election[\IDMVBAtilde]$ associated with an identity $\IDMVBAtilde$, satisfying the following properties:  
 \begin{itemize}
\item   {\bf Termination:} If $t+1$ honest nodes activate $\Election[\IDMVBAtilde]$, then each honest node that activates it will output a common value $\electionoutput$.     
\item  {\bf Consistency:} If any two honest nodes output $\electionoutput'$ and $\electionoutput''$ from $\Election[\IDMVBAtilde]$, respectively, then  $\electionoutput'=\electionoutput''$.
\item   {\bf Uniform:} The output $\electionoutput$ from $\Election[\IDMVBAtilde]$ is randomly generated based on a uniform distribution for $\electionoutput \in [1:n]$.
\item   {\bf Unpredictability:} The adversary cannot correctly predict the output of $\Election[\IDMVBAtilde]$ unless at least one honest node has activated it.
\end{itemize} 
\end{definition}

\noindent   {\bf Erasure code  ($\EC$).}    An $(n, k)$ erasure coding  scheme  consists of an encoding scheme $\ECEnc: \Alphabet^{k} \to  \Alphabet^{n}$ and a decoding scheme $\ECDec: \Alphabet^{k} \to  \Alphabet^{k}$, where $\Alphabet$ denotes the alphabet of each symbol and $\alphabetsize\defeq|\Alphabet|$  denotes the size of $\Alphabet$. With an $(n, k)$ erasure code, the original message can be decoded from any $k$ encoded symbols. Specifically, given  $[\EncodedSymbol_1,  \EncodedSymbol_2, \cdots, \EncodedSymbol_{n}] \gets \ECEnc(n, k, \MVBAInputMsg)$, then $\ECDec(n,k, \{\EncodedSymbol_{j_1}, \EncodedSymbol_{j_2}, \cdots  \EncodedSymbol_{j_k}\}) =\MVBAInputMsg$ holds true for any $k$ distinct integers $j_1, j_2, \cdots, j_k \in [1:n]$.

\vspace{5pt}

\begin{algorithm}
\caption{$\OciorABAstar$  protocol,  with an identifier $\IDABA$,  for $n\geq 3t+1$. Code is shown for $\Node_{\thisnodeindex}$.}  \label{algm:OciorABAstar}
\begin{algorithmic}[1]
\vspace{5pt}    
\footnotesize

 \Statex   \emph{//   **   $\RBC[\ltuple \IDABA, j \rtuple]$ denotes a    reliable broadcast ($\RBC$) instance, where Node~$j$ is the leader,  which calls  $\OciorRBC$ protocol \cite{ChenOciorCOOL:24} **}    
 \Statex   \emph{//   **   $\ABBA[\ltuple \IDABA, j \rtuple]$ denotes the $j$th instance of asynchronous  binary Byzantine agreement  ($\ABBA$), for $j\in [1:n]$   **}    
 
 \Statex
  \State  Initially set     $\APBVAinput_i[j]\gets \missing $, $\forall j\in [1:n]$ 
\State {\bf upon} receiving input  $\Me_{i}$  {\bf do}:  
\Indent  
	
	\State  $[y_1^{(\thisnodeindex)}, y_2^{(\thisnodeindex)}, \cdots, y_{n}^{(\thisnodeindex)}] \gets \ECEnc(\networksizen,   t    +1 , \Me_{\thisnodeindex})$     \label{line:abaECncoding}    \quad \quad \quad \quad\quad\quad\quad\quad\quad\quad\quad\quad  \emph{// erasure code encoding }
	
	\State  $\Pass$  $y_{\thisnodeindex}^{(\thisnodeindex)}$ into $\RBC[\ltuple \IDABA, \thisnodeindex \rtuple]$ as an input, where Node~$\thisnodeindex$ is the leader

\EndIndent

\State {\bf upon} delivery of $y_j^{(j)}$  from   $\RBC[\ltuple \IDABA, j \rtuple]$  {\bf do}:    \label{line:abainputabbacondition}     
\Indent  

	\State $\wait$ until $y_j^{(\thisnodeindex)}$ has been computed  \quad \quad \quad \quad \quad \quad \quad \quad \quad \quad \quad \quad   \quad \quad \quad \quad \emph{//     wait for executing Line~\ref{line:abaECncoding}}
    	\IfThenElse {$y_j^{(j)} = y_j^{(\thisnodeindex)}$}  {set $\APBVAinput_i[j] \gets 1$} {set $\APBVAinput_i[j] \gets 0$}    \label{line:abainputabbaconditionequalstar}    
	\State  $\Pass$  $\APBVAinput_i[j]$ into $\ABBA[\ltuple \IDABA, j \rtuple]$ as an input       \label{line:abainputabba}      \quad  \quad\quad   \quad \quad \quad  \quad\quad \quad \quad \quad  \quad\quad \emph{//       asynchronous  binary Byzantine agreement    }

\EndIndent

\State {\bf upon} delivery of outputs from  $n-t$   instances of $\ABBA$    {\bf do}:       \label{line:abaoutputabba} 
\Indent  

 	\State  $\Pass$ input $0$  to each instance of  $\ABBA$  that has not yet received input from this node

\EndIndent

\State {\bf upon} delivery of outputs from  all $n$   instances of $\ABBA$    {\bf do}:      \label{line:abanABBA} 
\Indent

	\State  let $\ABAOneSet  \subseteq [1:n]$ denote the indices of all $\ABBA$ instances that delivered $1$
	\If {   $|\ABAOneSet| < t+1$}        \label{line:abaLESSt1} 
		\State   $\Output$  $\defaultvalue$ and $\terminate$ 	  
	\Else
		\State  let $\ABAOneSetK  \subseteq \ABAOneSet$ denote the first $t+1$ smallest  values in   $\ABAOneSet$      \label{line:abaGEQt1a} 
 		\State $\wait$ for the  delivery of $y_j^{(j)}$  from   $\RBC[\ltuple \IDABA, j \rtuple]$, $\forall j\in \ABAOneSetK$  
		\State   $\hat{\Me} \gets \ECDec(n, t+1, \{y_j^{(j)}\}_{j\in \ABAOneSetK})$      \label{line:abaECdec}   \quad \quad \quad \quad\quad\quad\quad\quad\quad\quad\quad\quad\quad\quad\quad\quad \quad\quad\quad\quad \emph{// erasure code decoding }
		\State $\Output$  $\hat{\Me}$ and $\terminate$ 	    \label{line:abaGEQt1b} 
	\EndIf

\EndIndent

\end{algorithmic}
\end{algorithm}

\begin{figure} [H]
\centering
\includegraphics[width=18cm]{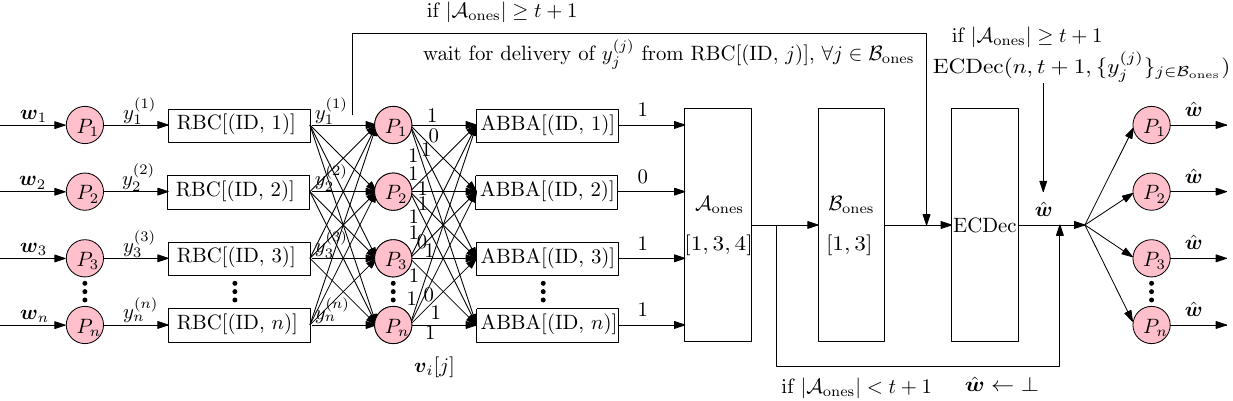}
\caption{A block diagram of the proposed $\OciorABAstar$ protocol with an identifier $\IDMVBA$. Here $\ABAOneSet  \subseteq [1:n]$ denotes the indices of all $\ABBA$ instances that delivered $1$, while $\ABAOneSetK  \subseteq \ABAOneSet$ denotes the first $t+1$ smallest  values in   $\ABAOneSet$. The description focuses on the example with $n=4$ and $t=1$. 
}
\label{fig:OciorABAstar}
\end{figure}

\section{$\OciorABAstar$}    \label{sec:OciorABAstar}

This proposed $\OciorABAstar$ is an error-free,  information-theoretically secure asynchronous $\BA$  protocol.   $\OciorABAstar$ does not rely on any   cryptographic assumptions, such as  signatures or hashing,  except for  the common coin assumption. 
This protocol achieves the asynchronous $\BA$ consensus on an $\ell$-bit message with an expected communication complexity of $O(n\ell + n^3 \log \alphabetsize )$  bits  and an expected  round complexity of $O(\log n)$ rounds,  under the optimal resilience  condition $n \geq 3t + 1$.  Here $\alphabetsize$ denotes the alphabet size of  the    error correction code used in the   protocol, inherited from the reliable broadcast protocols invoked within the protocol. 
When  error correction codes with a constant alphabet size are used, $\alphabetsize$ becomes a constant.

\subsection{Overview of $\OciorABAstar$}     
The proposed  $\OciorABAstar$ is described in Algorithm~\ref{algm:OciorABAstar}. Fig.~\ref{fig:OciorABAstar} presents a block diagram of the proposed $\OciorABAstar$ protocol.  The main steps of $\OciorABAstar$ are outlined below.  
\begin{itemize}
\item    Erasure code  encoding: Each honest Node~$\thisnodeindex$ first encodes its initial message into $n$ symbols, $y_1^{(\thisnodeindex)}, y_2^{(\thisnodeindex)}, \cdots, y_{n}^{(\thisnodeindex)}$,  using an erasure code. 
\item    $\RBC$:  Then Node~$\thisnodeindex$ passes   $y_{\thisnodeindex}^{(\thisnodeindex)}$ into $\RBC[\ltuple \IDABA, \thisnodeindex \rtuple]$ as an input. 
    $\RBC[\ltuple \IDABA, \thisnodeindex \rtuple]$ represents  a   reliable broadcast  instance where Node~$\thisnodeindex$ is the leader.  $\RBC[\ltuple \IDABA, \thisnodeindex \rtuple]$ calls  the $\OciorRBC$ protocol \cite{ChenOciorCOOL:24}. 
\item    Set input values for $\ABBA$:   Upon  delivery of $y_j^{(j)}$  from   $\RBC[\ltuple \IDABA, j \rtuple]$, Node~$\thisnodeindex$ sets   $\APBVAinput_i[j] \gets 1$ if $y_j^{(j)} = y_j^{(\thisnodeindex)}$; otherwise,  Node~$\thisnodeindex$ sets  $\APBVAinput_i[j] \gets 0$.  Then, Node~$\thisnodeindex$ passes   $\APBVAinput_i[j]$ into $\ABBA[\ltuple \IDABA, j \rtuple]$ as an input, where   $\ABBA[\ltuple \IDABA, j \rtuple]$ denotes the $j$th instance of asynchronous  binary Byzantine agreement  ($\ABBA$), for $j\in [1:n]$. 
\item    Set remaining input values for $\ABBA$:   Upon   delivery of outputs from  $n-t$   instances of $\ABBA$,  Node~$\thisnodeindex$ passes   input $0$  to each instance of  $\ABBA$  that has not yet received input from this node.  
\item     Erasure code  decoding:   Upon  delivery of outputs from  all $n$   instances of $\ABBA$, let $\ABAOneSet  \subseteq [1:n]$ denote the indices of all $\ABBA$ instances that delivered $1$.
\begin{itemize}
\item  If $|\ABAOneSet| < t+1$, then each honest node outputs   $\defaultvalue$ and terminates. 
\item  If $|\ABAOneSet| \geq  t+1$,    let $\ABAOneSetK  \subseteq \ABAOneSet$ denote the first $t+1$ smallest  values in   $\ABAOneSet$.
In this case, each honest node waits for the  delivery of $y_j^{(j)}$  from   $\RBC[\ltuple \IDABA, j \rtuple]$, $\forall j\in \ABAOneSetK$, and then decodes  
	   $\hat{\Me} \gets \ECDec(n, t+1, \{y_j^{(j)}\}_{j\in \ABAOneSetK})$      using erasure code decoding. Each honest node     then 
outputs  $\hat{\Me}$ and terminates. 
\end{itemize} 
\end{itemize}

\subsection{Analysis of $\OciorABAstar$}     

The main results of  $\OciorABAstar$ are summarized in the following theorems.

 \begin{theorem}  [Termination]  \label{thm:OciorABAterminate}
Given $n\geq 3t+1$,  each  honest node  eventually  outputs a  message and terminate   in $\OciorABAstar$.  
\end{theorem}
\begin{proof}
From Lemma~\ref{lm:OciorABApropertyABBAn}, it is concluded that all $n$ instances of  $\ABBA$  eventually deliver outputs at each honest node.   
Based on this conclusion,  the condition in Line~\ref{line:abanABBA} of Algorithm~\ref{algm:OciorABAstar} is eventually triggered at each honest node.  At this point, if the condition in  Line~\ref{line:abaLESSt1} of Algorithm~\ref{algm:OciorABAstar} is satisfied, i.e.,    $|\ABAOneSet| < t+1$,  then each honest node eventually outputs  $\defaultvalue$ and terminates,   where $\ABAOneSet  \subseteq [1:n]$ denote the indices of all $\ABBA$ instances that delivered $1$.  If  $|\ABAOneSet| \geq t+1$,  each honest node eventually runs the steps in Lines~\ref{line:abaGEQt1a}-\ref{line:abaGEQt1b} and terminates.  It is worth noting that, based on Lemma~\ref{lm:OciorABApropertyABBARBC},   $t+1$ instances of   $\{\RBC[\ltuple \IDABA, j \rtuple]\}_{j\in \ABAOneSetK}$ eventually deliver outputs at  each honest node,  
where $\ABAOneSetK  \subseteq \ABAOneSet$ denotes the first $t+1$ smallest  values in   $\ABAOneSet$. 
 \end{proof}

 \begin{theorem}  [Validity]  \label{thm:OciorABAvalidity}
Given $n\geq 3t+1$, if  all  honest nodes input the same  value $\wv$, then each honest node eventually outputs $\wv$ in $\OciorABAstar$.   
\end{theorem}
\begin{proof}
From Lemma~\ref{lm:OciorABAvalidity}, if  all  honest nodes input the same  value $\wv$, then   at least $t+1$ instances of $\ABBA$  eventually output $1$.  
Based on this results and Lemma~\ref{lm:OciorABApropertyABBAn}, if  all  honest nodes input the same  value $\wv$, every honest node eventually runs the steps in Lines~\ref{line:abaGEQt1a}-\ref{line:abaGEQt1b} of Algorithm~\ref{algm:OciorABAstar}.   
From Lemma~\ref{lm:OciorABApropertyABBAy}, if  all  honest nodes input the same  value $\wv$ and  $\ABBA[\ltuple \IDABA, j \rtuple]$  delivers an output $1$,  then  $\RBC[\ltuple \IDABA, j \rtuple]$   eventually delivers the same output  $\ECEnc_j(\networksizen,   t +1 , \wv)$ at each honest node, for $j\in \ABAOneSetK$, where $\ABAOneSetK$   denotes the first $t+1$ smallest  indices of all $\ABBA$ instances that delivered $1$. In this case, every  honest node eventually output the same the decoded message $\wv$ (see  Line~\ref{line:abaECdec} of Algorithm~\ref{algm:OciorABAstar}), where  $\wv=\ECDec(n, t+1, \{\ECEnc_j(\networksizen,   t +1 , \wv)\}_{j\in \ABAOneSetK})$. 
\end{proof}

\begin{theorem}  [Consistency]   \label{thm:OciorABAconsistency} 
Given $n\geq 3t+1$,  if any honest node outputs a value $\wv$, then every honest node eventually outputs $\wv$   in $\OciorABAstar$. 
\end{theorem}
\begin{proof}
From Lemma~\ref{lm:OciorABApropertyABBAn},  all $n$ instances of  $\ABBA$  eventually deliver outputs at each honest node.   
Thus,  the condition in Line~\ref{line:abanABBA} of  Algorithm~\ref{algm:OciorABAstar} is eventually triggered at each honest node.  If the condition in  Line~\ref{line:abaLESSt1} is satisfied,  then every honest node eventually outputs  $\defaultvalue$.  Otherwise,  every honest node eventually runs the steps in Lines~\ref{line:abaGEQt1a}-\ref{line:abaGEQt1b} and  outputs the same the decoded message, i.e.,  $\ECDec(n, t+1, \{y_j^{(j)}\}_{j\in \ABAOneSetK})$  for some $\{y_j^{(j)}\}_{j\in \ABAOneSetK}$ (see  Line~\ref{line:abaECdec}).  
It is worth noting   from Lemma~\ref{lm:OciorABApropertyABBARBC} that,   $\RBC[\ltuple \IDABA, j \rtuple]$   eventually delivers the same output  $y_j^{(j)}$ at each honest node,   $\forall j\in \ABAOneSetK$.  
 \end{proof}

 \begin{lemma}  [Property~$1$ of $\ABBA$]  \label{lm:OciorABApropertyABBARBC}
In $\OciorABAstar$, if $\ABBA[\ltuple \IDABA, j \rtuple]$  delivers an output $1$ at one honest node,  then $\RBC[\ltuple \IDABA, j \rtuple]$   eventually delivers the same output  $y_j^{(j)}$ at each honest node, for some $y_j^{(j)}$, for $j\in \ABAOneSetK$, where $\ABAOneSetK$   denotes the first $t+1$ smallest  indices of all $\ABBA$ instances that delivered $1$. 
\end{lemma}
\begin{proof}
In $\OciorABAstar$, if $\ABBA[\ltuple \IDABA, j \rtuple]$  delivers an output $1$ at one honest node, for $j\in \ABAOneSetK$,  then at least one honest node has provided an input $1$ to $\ABBA[\ltuple \IDABA, j \rtuple]$  (otherwise, $\ABBA[\ltuple \IDABA, j \rtuple]$ will deliver an output $0$). In this case, the honest node providing an input $1$ to $\ABBA[\ltuple \IDABA, j \rtuple]$ should have received the output $y_j^{(j)}$  from   $\RBC[\ltuple \IDABA, j \rtuple]$  for some $y_j^{(j)}$ (see Lines~\ref{line:abainputabbacondition}-\ref{line:abainputabba} of Algorithm~\ref{algm:OciorABAstar}).  From the Totality and Consistency properties of the  $\RBC$, if one  honest node outputs a value $y_j^{(j)}$, then every honest node  eventually outputs the same value $y_j^{(j)}$.   
 \end{proof}

 \begin{lemma}  [Property~$2$ of $\ABBA$]  \label{lm:OciorABApropertyABBAy}
In $\OciorABAstar$, if $\ABBA[\ltuple \IDABA, j \rtuple]$  delivers an output $1$ at one honest node, for $j\in \ABAOneSetK$,    then  $\RBC[\ltuple \IDABA, j \rtuple]$   eventually delivers the same output  $y_j^{(j)}$ at each honest node. Furthermore,  there exists  an honest Node~$i$, $i\in [1:n]$, with an initial message $\Me_{i}$ such that its $j$th encoded symbol    is equal to $y_j^{(j)}$, i.e.,   $\ECEnc_j(\networksizen,   t +1 , \Me_{i}) = y_j^{(j)}$, where $\ECEnc_j( )$ denotes the $j$th encoded symbol.
\end{lemma}
\begin{proof}
From the result of Lemma~\ref{lm:OciorABApropertyABBARBC}, if $\ABBA[\ltuple \IDABA, j \rtuple]$  delivers an output $1$ at one honest node, for $j\in \ABAOneSetK$,   then $\RBC[\ltuple \IDABA, j \rtuple]$   eventually delivers the same output  $y_j^{(j)}$ at each honest node, for some $y_j^{(j)}$.  Furthermore, if $\ABBA[\ltuple \IDABA, j \rtuple]$  delivers an output $1$ at one honest node,   then at least one honest node has provided an input $1$ to $\ABBA[\ltuple \IDABA, j \rtuple]$.  If an honest node, say, Node~$i$, $i\in [1:n]$,  provides an input $1$ to $\ABBA[\ltuple \IDABA, j \rtuple]$, then its $j$th encoded symbol  should be equal to $y_j^{(j)}$, i.e., $\ECEnc_j(\networksizen,   t +1 , \Me_{i}) = y_j^{(j)}$ (see Line~\ref{line:abainputabbaconditionequalstar} of Algorithm~\ref{algm:OciorABAstar}). 
 \end{proof}

  \begin{lemma}  [Property~$3$ of $\ABBA$]  \label{lm:OciorABApropertyABBAn}
In $\OciorABAstar$,  all $n$ instances of  $\ABBA$  eventually deliver outputs at each honest node.  
\end{lemma}
\begin{proof}
At first we will argue that there is at least one honest node delivering complete outputs from all $n$ instances of  $\ABBA$ in $\OciorABAstar$.
The argument is based on  proof by contradiction. Let us  assume that there is no honest node delivering complete outputs from all $n$ instances of $\ABBA$.  Under this assumption,   at least $n-t$ instances of  $\RBC$  eventually deliver outputs at each honest node, resulting from the Validity property of $\RBC$.  Then, all honest nodes eventually make inputs to at least  $n-t$ instances of $\ABBA$ (see Line~\ref{line:abainputabba} of Algorithm~\ref{algm:OciorABAstar}).  It   implies that at least   $n-t$ instances of $\ABBA$  eventually deliver outputs at each honest node.  Thus, the condition in Line~\ref{line:abaoutputabba} of Algorithm~\ref{algm:OciorABAstar} is eventually triggered and all honest nodes eventually make inputs to each of   $n$ instances of $\ABBA$, which suggests that all $n$ instances of  $\ABBA$  eventually deliver outputs at each honest node, resulting from the Termination property of Byzantine agreement.   This result contradicts with the assumption. Therefore, there is at least one honest node delivering complete outputs from all $n$ instances of  $\ABBA$ in $\OciorABAstar$.

Since at least one honest node delivers complete outputs from all $n$ instances of  $\ABBA$, then from the Consistency  property of Byzantine agreement it is guaranteed that every honest node eventually delivers   complete outputs from all $n$ instances of  $\ABBA$. 
\end{proof}

   \begin{lemma}  [Property~$4$ of $\ABBA$]  \label{lm:OciorABAvalidity}
In $\OciorABAstar$,  if  all  honest nodes input the same  value $\wv$, then   at least $t+1$ instances of $\ABBA$  eventually output $1$.  
\end{lemma}
\begin{proof}
 If all honest nodes input the same value $\wv$, then before the condition in Line~\ref{line:abaoutputabba} of Algorithm~\ref{algm:OciorABAstar} is triggered for the first time among all honest nodes,  there is no honest node providing an input of $0$ to any instance of $\ABBA[\ltuple \IDABA, j \rtuple]$, $\forall j\in [1:n]\setminus \Fc$  (see Line~\ref{line:abainputabbaconditionequalstar} of Algorithm~\ref{algm:OciorABAstar}).  
Therefore, when  the condition in Line~\ref{line:abaoutputabba} of Algorithm~\ref{algm:OciorABAstar} is triggered for the first time among all honest nodes, at least   $n-t -t\geq t+1$   instances of $\ABBA$ have delivered  output $1$, based on the Validity property of Byzantine agreement.  
\end{proof}

\begin{algorithm}
\caption{$\OciorABA$  protocol,  with an identifier $\IDABA$,  for $n\geq 3t+1$. Code is shown for $\Node_{\thisnodeindex}$.}  \label{algm:OciorABA}
\begin{algorithmic}[1]
\vspace{5pt}    
\footnotesize

 \Statex   \emph{//   **   $\RBC[\ltuple \IDABA, j \rtuple]$ denotes a    reliable broadcast ($\RBC$) instance, where Node~$j$ is the leader,  which calls  $\OciorRBC$ protocol \cite{ChenOciorCOOL:24} **}   
 \Statex   \emph{//   **    Let us define $\missing:=2 $  **}     

\Statex

  \State  Initially set     $\APBVAinput_i[j]\gets \missing $, $\forall j\in [1:n]$

\State {\bf upon} receiving input  $\Me_{i}$  {\bf do}:  
\Indent  
	
	\State  $[y_1^{(\thisnodeindex)}, y_2^{(\thisnodeindex)}, \cdots, y_{n}^{(\thisnodeindex)}] \gets \ECEnc(\networksizen,   t    +1 , \Me_{\thisnodeindex})$     \label{line:OciorABAabaECncoding}    \quad \quad \quad \quad\quad\quad\quad\quad\quad\quad\quad\quad  \emph{// erasure code encoding }
	
	\State  $\Pass$  $y_{\thisnodeindex}^{(\thisnodeindex)}$ into $\RBC[\ltuple \IDABA, \thisnodeindex \rtuple]$ as an input, where Node~$\thisnodeindex$ is the leader

\EndIndent

\State {\bf upon} delivery of $y_j^{(j)}$  from   $\RBC[\ltuple \IDABA, j \rtuple]$, for $j\in [1:n]$  {\bf do}:    \label{line:OciorABAabainputabbacondition}     
\Indent  

	\State $\wait$ until $y_j^{(\thisnodeindex)}$ has been computed  \quad \quad \quad \quad \quad \quad \quad \quad \quad \quad \quad \quad   \quad \quad \quad \quad \emph{//     wait for executing Line~\ref{line:OciorABAabaECncoding}}
    	\IfThenElse {$y_j^{(j)} = y_j^{(\thisnodeindex)}$}  {set $\APBVAinput_i[j] \gets 1$} {set $\APBVAinput_i[j] \gets 0$}    \label{line:abainputabbaconditionequal}    
	\State  $\Pass$  $\APBVAinput_i[j]$ into $\APVA[\IDABA]$ as an input       \label{line:OciorABAabainputabba}      \quad  \quad\quad   \quad \quad \quad  \quad\quad \quad \quad \quad  \quad\quad \emph{//      asynchronous partial vector agreement ($\APVA$)     }

\EndIndent

\State {\bf upon}  $\APVA[\IDABA]$  outputs $\APBVAoutput$   {\bf do}:  	         \label{line:OciorABAAPVAtrigger} 
 
\Indent  
	\State  let $\ABAOneSet  \subseteq [1:n]$ denote the indices of all elements in $\APBVAoutput$    that  are  equal to  $1$, i.e., $\ABAOneSet=\{j:  \APBVAoutput[j] = 1, j \in[1:n]\}$
	\If {   $|\ABAOneSet| < t+1$}        \label{line:OciorABAabaLESSt1} 
		\State   $\Output$  $\defaultvalue$ and $\terminate$ 	  
	\Else
		\State  let $\ABAOneSetK  \subseteq \ABAOneSet$ denote the first $t+1$ smallest    values in   $\ABAOneSet$      \label{line:OciorABAabaGEQt1a} 
 		\State $\wait$ for the  delivery of $y_j^{(j)}$  from   $\RBC[\ltuple \IDABA, j \rtuple]$, $\forall j\in \ABAOneSetK$         \label{line:OciorABARBCoutputs} 
		\State   $\hat{\Me} \gets \ECDec(n, t+1, \{y_j^{(j)}\}_{j\in \ABAOneSetK})$      \label{line:OciorABAabaECdec}   \quad \quad \quad \quad\quad\quad\quad\quad\quad\quad\quad\quad\quad\quad\quad\quad \quad\quad\quad\quad \emph{// erasure code decoding }
		\State $\Output$  $\hat{\Me}$ and $\terminate$ 	    \label{line:OciorABAabaGEQt1b} 
	\EndIf

\EndIndent

\end{algorithmic}
\end{algorithm}

\begin{algorithm}  
\caption{$\APVA$  protocol, with identifier $\IDMVBA$, for $n\geq 3t+1$. Code is shown for $\Node_{\thisnodeindex}$.}    \label{algm:APVA} 
\begin{algorithmic}[1]
\vspace{5pt}    
\footnotesize

 \Statex   \emph{//   **    $\IDABANew :=f(\IDABA)$ represents a new identity generated as a function of   $\IDABA$;    Let us define $\missing:=2 $ **}    
   \Statex   \emph{//   **    $\NonmissingElementSet(\Me)$ denotes the set of indices of all non-missing  elements of  $\Me$,  i.e., $\NonmissingElementSet(\Me):=\{j:  \Me[j] \neq \missing, j \in[1:n]\}$.  **}

 \State {\bf upon} receiving   input    $\APBVAinput_i[j] \in  \{1,0\}$  for $j\in [1:n]$ {\bf do}:      
\Indent 
 	\State  $\Pass$  $\APBVAinput_i[j]$ into $\ACIDd[\IDABA]$ as an input       
 \EndIndent	
 
\State {\bf upon} delivery of output $[ \ReadyRecord^{(1)}, \ReadyRecord^{(0)},  \FinishRecord^{(1)},  \FinishRecord^{(0)}, \RBCReadyindicator,  \RBCFinishindicator]$	  from    $\ACIDd[ \IDMVBA ]$  {\bf do}:       \label{line:APVAACIDoutput}  
\Indent  
	\For {$\electionround \inset [1:n]$}   \label{line:APVAABAround} 
	
		\State  $\electionoutput \gets \Election[\ltuple \IDMVBA,  \electionround \rtuple]$    \label{line:APVAElection}     \quad   \quad \quad \quad \quad 	\quad \quad \quad   \quad \quad \quad  \quad   \emph{//      an election protocol } 
		\State  $\ABBAoutputNew \gets \ABBBA[\ltuple \IDABANew,   \electionoutput, 0\rtuple](\RBCReadyindicator[\electionoutput], \RBCFinishindicator[\electionoutput])$    \label{line:APVAABBBA}         \  \quad      \emph{//      asynchronous  biased binary BA   ($\ABBBA$, calling Algorithm~2 of \cite{ChenOciorMVBA:24})   } 
		\State  $\ABAoutputNew\gets \ABBA[\ltuple \IDABANew,  \electionoutput\rtuple](\ABBAoutputNew)$   \label{line:APVAABBA}                  \quad \quad  \quad\quad \quad \quad \quad  \quad\quad \emph{//     an asynchronous  binary BA  ($\ABBA$)   } 
		
		\If {$\ABAoutputNew \eqlog 1$}     \label{line:APVAABAoutput1} 
 				\State $\wait$ for $\RBC[\ltuple \IDABANew, \electionoutput \rtuple]$ to output a vector $\ConfirmRecord_{\electionoutput}$	          \label{line:APVARBCout} 
 				\If {$|\NonmissingElementSet(\ConfirmRecord_{\electionoutput})| \geq n-t$	}        \label{line:APVANNMissingThreshold} 
 					\State $\ABBAoutput_{\electionoutput, j} \gets \ABBBA[\ltuple \IDMVBA,  \electionoutput, j\rtuple] (\ReadyRecord^{(\ConfirmRecord_{\electionoutput}[j])}[j], \FinishRecord^{(\ConfirmRecord_{\electionoutput}[j])}[j])$,  $\forall j\in \NonmissingElementSet(\ConfirmRecord_{\electionoutput})$	      \label{line:APVAABBBAinput}          \quad    \emph{//    $n-t$ parallel   $\ABBBA$ }  
 					\IfThenElse { $\sum_{j\in \NonmissingElementSet(\ConfirmRecord_{\electionoutput})} \ABBAoutput_{\electionoutput, j}  \eqlog |\Mc(\ConfirmRecord_{\electionoutput})|$}  {set $\ABBAoutput' \gets 1$} {set $\ABBAoutput' \gets 0$}          \label{line:APVAABBBAoutput}     
 					\State  $\ABAoutput'\gets \ABBA[\ltuple \IDMVBA,  \electionoutput\rtuple](\ABBAoutput')$    \label{line:APVAABBA2}   
					\If {$\ABAoutput'=1$}     \label{line:APVAoutCondition}
						 \State     $\Output$  $\ConfirmRecord_{\electionoutput}$  and $\terminate$ 		      \label{line:APVAout} 
					\EndIf
					
				\EndIf

		\EndIf

	\EndFor

\EndIndent

\State {\bf upon}   delivery of   $\ConfirmRecord_{\thisnodeindex}$	 from  $\ACIDd[ \IDMVBA ]$    {\bf do}:  
\Indent  
 	 \State  $\Pass$  $\ConfirmRecord_{\thisnodeindex}$  into $\RBC[\ltuple \IDABANew, \thisnodeindex \rtuple]$  as an input        \label{line:APVARBCinputthisnode} 
\EndIndent

\State {\bf upon}   delivery of output $\ConfirmRecord_{j}$	 from   $\RBC[\ltuple \IDABANew, j \rtuple]$ for $j\in [1:n]$     {\bf do}:  
\Indent  
 	 \State  $\deliver$    $\ConfirmRecord_{j}$ into $\ACIDd[ \IDMVBA ]$ as an input        \label{line:APVARBCoutput} 
\EndIndent

\end{algorithmic}
\end{algorithm}

\begin{algorithm}
\caption{$\ACIDd$  protocol with identifier $\IDABA$. Code is shown for $\Node_{\thisnodeindex}$.}  \label{algm:ACIDd}
\begin{algorithmic}[1]
\vspace{5pt}    
\footnotesize

 \Statex   \emph{//   **    $\IDABANew :=f(\IDABA)$ represents a new identity generated as a function of   $\IDABA$ **}    
  \Statex   \emph{//   **    Let us define $\missing:=2 $  **}    
  
 \Statex

 \State  Initially set  $\ReadyRecord^{(\binaryvalue)}[j]\gets 0$,   $\FinishRecord^{(\binaryvalue)}[j]\gets 0$, $\RBCReadyindicator[j]\gets 0$,  $\RBCFinishindicator [j]\gets 0$,  $\forall j\in [1:n], \forall \binaryvalue \in \{1,0\}$;  set  $\ConfirmRecord_{\thisnodeindex} [j]\gets \missing $, $\forall j\in [1:n]$;   and $\ConfirmCount\gets 0$

\Statex

\Statex   \emph{//   **   vote  **} 	 	
\State {\bf upon} receiving an input  value  $\APBVAinput_i[j] \in  \{1,0\}$, for $j\in[1:n]$   {\bf do}:       \label{line:ACIDdVoteBegin} 
\Indent  

            	 \State  $\send$   $\ltuple   \VOTE, \IDABA,  j ,  \APBVAinput_i[j] \rtuple$     to all nodes                                   
	  
\EndIndent

\Statex 

\Statex   \emph{//   ** ready **}  

\State {\bf upon} receiving  $t+1$  $\ltuple   \VOTE, \IDABA,  j , \binaryvalue  \rtuple$ messages from distinct nodes, for   the same  $\binaryvalue \in \{1,0\}$, and same $j\in[1:n]$ {\bf do}:  
\Indent  
	\If {$\ltuple   \VOTE, \IDABA,  j , \binaryvalue \rtuple$   not yet sent }     
		\State $\send$ $\ltuple   \VOTE, \IDABA,  j , \binaryvalue \rtuple$ to  all nodes       	
	\EndIf
	
	\State set $\ReadyRecord^{(\binaryvalue)}[j] \gets 1$                                                      \label{line:ACIDdReady} 
	\State $\send$ $\ltuple   \READY, \IDABA,  j , \binaryvalue  \rtuple$    to  all nodes         
\EndIndent

\Statex

\Statex   \emph{//   ** finish **}  

\State {\bf upon} receiving  $n-t$  $\ltuple   \READY, \IDABA,  j , \binaryvalue  \rtuple$  messages from distinct nodes,   for   the same  $\binaryvalue \in \{1,0\}$, and same $j\in[1:n]$ {\bf do}:  
\Indent  
	\State set $\FinishRecord^{(\binaryvalue)}[j] \gets 1$          \label{line:ACIDdFinish} 
	\State $\send$ $\ltuple   \FINISH, \IDABA,  j , \binaryvalue  \rtuple$    to  all nodes              \label{line:ACIDdFinishSend} 
\EndIndent

\Statex

\Statex   \emph{//   **   confirm **}

\State {\bf upon} receiving   $n-t$  $\ltuple   \FINISH, \IDABA,  j , \binaryvalue  \rtuple$  messages from distinct nodes, for   the same  $\binaryvalue \in \{1,0\}$, and same $j\in[1:n]$ {\bf do}:     \label{line:ACIDdConfirmCond} 
\Indent  
	\If {$\ConfirmRecord_{\thisnodeindex}[j] = \missing$ }     
		\State set $\ConfirmRecord_{\thisnodeindex}[j] \gets \binaryvalue$; $\ConfirmCount\gets \ConfirmCount+1$      \label{line:ACIDdConfirm}   
	\EndIf
 
\EndIndent

\Statex

\Statex   \emph{//   ** reliable broadcast of $\ConfirmRecord_{\thisnodeindex}$**}

\State {\bf upon}    $\ConfirmCount= n-t$  {\bf do}:        \label{line:ACIDdcntCond}
\Indent  
 	 \State  $\deliver$    $\ConfirmRecord_{\thisnodeindex}$             \quad \quad   \emph{//      $\ConfirmRecord_{\thisnodeindex}$ will be passed into $\RBC[\ltuple \IDABANew, \thisnodeindex \rtuple]$ as an input   (see Line~\ref{line:APVARBCinputthisnode} of Algorithm~\ref{algm:APVA}) }

\EndIndent

\Statex

\Statex   \emph{//   ** $\RBC$-ready **}  

\State {\bf upon} receiving an input  value  $\ConfirmRecord_{j}$, for $j\in[1:n]$   {\bf do}:      \label{line:ACIDdReadyRBCCond}      \quad \quad   \emph{//   $\ConfirmRecord_{j}$ is  the  output from   $\RBC[\ltuple \IDABANew, j \rtuple]$ (see Line~\ref{line:APVARBCoutput} of Algorithm~\ref{algm:APVA})}        
\Indent
 	\State set $\RBCReadyindicator   [j] \gets 1$             \label{line:ACIDdReadyRBC} 
	\State $\send$ $\ltuple   \READY, \IDABANew,  j   \rtuple$    to  all nodes         
\EndIndent

\Statex

\Statex   \emph{//   ** $\RBC$- finish **}  

\State {\bf upon} receiving  $n-t$  $\ltuple   \READY, \IDABANew,  j   \rtuple$  messages from distinct nodes for the same  $j\in[1:n]$ {\bf do}:  
\Indent  
 	\State set $\RBCFinishindicator   [j] \gets 1$          \label{line:ACIDdFinishRBC} 
	\State $\send$ $\ltuple   \FINISH, \IDABANew,  j   \rtuple$    to  all nodes         
\EndIndent

\Statex

\Statex   \emph{//   ** $\RBC$- confirm and vote for election**}

\State {\bf upon} receiving   $n-t$  $\ltuple   \FINISH, \IDABANew,  \thisnodeindex   \rtuple$   messages from distinct nodes {\bf do}:    \label{line:ACIDdElectionSendCond}
\Indent  
	\State $\send$ $(\ELECTION, \IDMVBA)$ to  all nodes    \label{line:ACIDdElectionSend}      \quad \quad   \emph{//      $\ACD[\ltuple \IDMVBA,  \thisnodeindex \rtuple]$ is complete at this point   } 	
 
\EndIndent

\Statex

\Statex   \emph{//   ** confirm for election **}

\State {\bf upon} receiving   $n-t$  $(\ELECTION, \IDABA)$  messages from distinct nodes and  $(\CONFIRM, \IDABA)$  not yet sent {\bf do}:      \label{line:ACIDdConfirmEleBegin} 
\Indent  
	\State $\send$ $(\CONFIRM, \IDABA)$ to  all nodes       	
 
\EndIndent

\State {\bf upon} receiving   $t+1$  $(\CONFIRM, \IDABA)$  messages from distinct nodes and  $(\CONFIRM, \IDABA)$  not yet sent {\bf do}:  
\Indent  
	\State $\send$ $(\CONFIRM, \IDABA)$ to  all nodes       	
 
\EndIndent

\Statex

\Statex   \emph{//   ** return and stop **} 
\State {\bf upon} receiving   $2t+1$  $(\CONFIRM, \IDABA)$  messages from distinct nodes {\bf do}:  
\Indent  
	\If {$(\CONFIRM, \IDABA)$   not yet sent }     
		\State $\send$ $(\CONFIRM, \IDABA)$ to  all nodes       	
	\EndIf
	\State  $\Return$  $[ \ReadyRecord^{(1)}, \ReadyRecord^{(0)},  \FinishRecord^{(1)},  \FinishRecord^{(0)}, \RBCReadyindicator,  \RBCFinishindicator]$	     \label{line:ACIDdConfirmEleEnd} 
\EndIndent

\end{algorithmic}
\end{algorithm}

\begin{algorithm} 
\caption{$\ABBBA$  protocol  (Algorithm~2 of \cite{ChenOciorMVBA:24}) with   $\IDMVBAtilde:=\ltuple \IDMVBA,  \electionoutput, j \rtuple$.  Code is shown for   $\Node_{\thisnodeindex}$.  }    \label{algm:ABBBA} 
\begin{algorithmic}[1]
\vspace{5pt}    
\footnotesize
 
\State {\bf upon} receiving  input  $(\abbainputA, \abbainputB)$, for some $\abbainputA, \abbainputB \in \{0,1\}$   {\bf do}:  		
\Indent
	\State $\ABBACountA \gets 0; \ABBACountB \gets 0; \ABBACountC \gets 0$   
	\State $\send$   $(\ABBAVALUE,  \IDMVBAtilde, \abbainputA,  \abbainputB)$   to all nodes
	
	\If {$(\abbainputA \eqlog 1) \OR (\abbainputB \eqlog 1)$ } $\Output$  $1$ and $\terminate$       \label{line:ABBBAoutput1} 
	\EndIf

	\State {\bf wait} for  at least one of the following events: 1) $\ABBACountA \geq \networkfaultsizet+1$,  2) $\ABBACountB \geq \networkfaultsizet+1$,  or 3)  $\ABBACountC \geq \networksizen-\networkfaultsizet$   
	\Indent
		\If {$(\ABBACountA \geq \networkfaultsizet+1) \OR (\ABBACountB \geq \networkfaultsizet+1)$}		    \label{line:ABBBAoneoutputCond} 
			\State  $\Output$  $1$ and $\terminate$        \label{line:ABBBAoneoutput} 
		\ElsIf {$\ABBACountC \geq \networksizen-\networkfaultsizet$}		     \label{line:zerocondition} 
			\State  $\Output$  $0$ and $\terminate$     \label{line:ABBBAzero} 
		\EndIf
	\EndIndent	
				
\EndIndent

\State {\bf upon} receiving  $(\ABBAVALUE, \IDMVBAtilde, \abbainputA,  \abbainputB)$ from  $\Node_j$ for the first time, for some $\abbainputA, \abbainputB \in \{0,1\}$ {\bf do}:  
\Indent  
	\State  $\ABBACountA \gets \ABBACountA +\abbainputA; \ABBACountB \gets \ABBACountB +\abbainputB$   
	\If {$\abbainputB \eqlog 0$ } $\ABBACountC \gets \ABBACountC +1$
	\EndIf			

\EndIndent

\end{algorithmic}
\end{algorithm}

\begin{figure} [b]
\centering
\includegraphics[width=18cm]{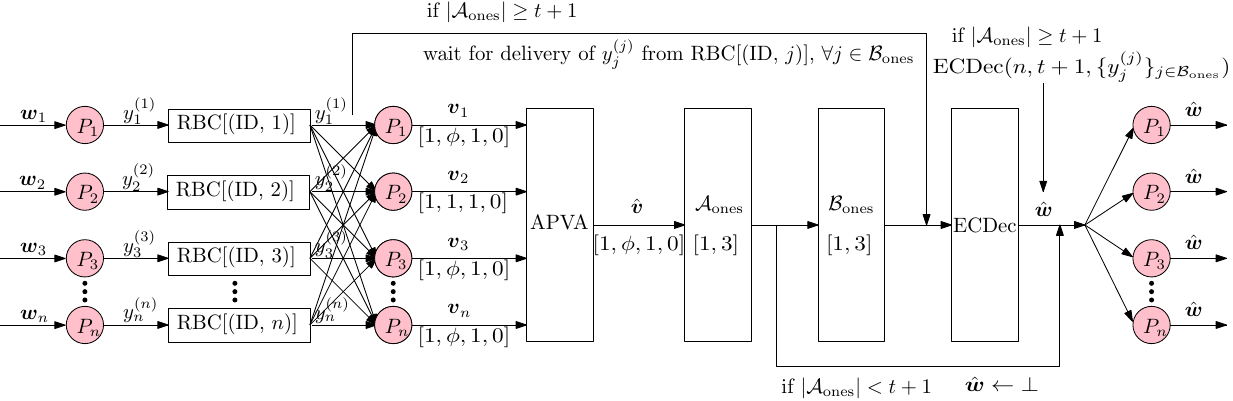}
\caption{A block diagram of the proposed $\OciorABA$ protocol with an identifier $\IDMVBA$. Here  $\APBVAoutput$ denotes the output vector of $\APVA$. 
$\ABAOneSet  \subseteq [1:n]$ denotes  the indices of all elements in  $\APBVAoutput$    that  are  equal to  $1$, i.e., $\ABAOneSet=\{j:  \APBVAoutput[j] = 1, j \in[1:n]\}$, while $\ABAOneSetK  \subseteq \ABAOneSet$ denotes the first $t+1$ smallest  values in   $\ABAOneSet$. The description focuses on the example with $n=4$ and $t=1$. 
}
\label{fig:OciorABA}
\end{figure}

\begin{figure} 
\centering
\includegraphics[width=18cm]{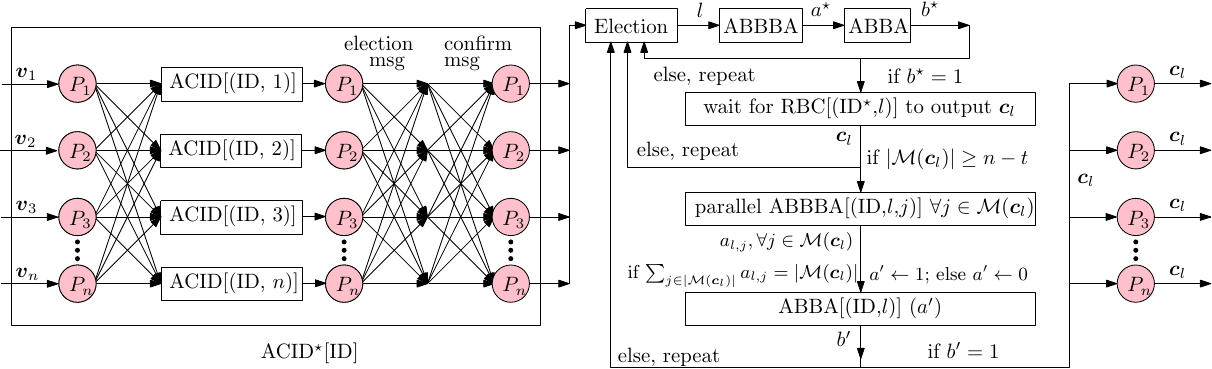}
\caption{A block diagram of the proposed $\APVA$ protocol with an identifier $\IDMVBA$.  Details of the $\ACIDd$ protocol are presented in Fig.~\ref{fig:ACIDstar}. Here $\NonmissingElementSet(\Me)$ denotes the set of indices of all non-missing  elements of the vector  $\Me$,  i.e., $\NonmissingElementSet(\Me):=\{j:  \Me[j] \neq \missing, j \in[1:n]\}$. 
}
\label{fig:APVA}
\end{figure}

\begin{figure} 
\centering
\includegraphics[width=17cm]{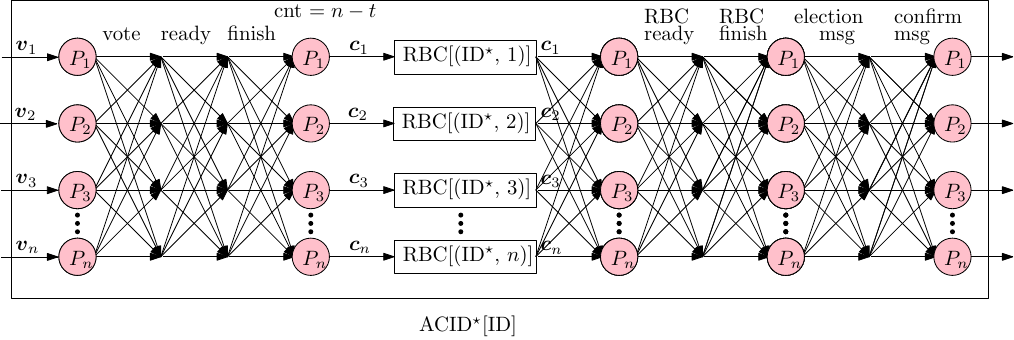}
\caption{A block diagram of the proposed $\ACIDd$ protocol with an identifier $\IDMVBA$. 
}
\label{fig:ACIDstar}
\end{figure}

\section{$\OciorABA$}    \label{sec:OciorABA}

This proposed $\OciorABA$ is an improved error-free,  information-theoretically secure asynchronous $\BA$  protocol.  $\OciorABA$ does not rely on any   cryptographic assumptions, such as  signatures or hashing,  except for  the common coin assumption. 
This protocol achieves the asynchronous $\BA$ consensus on an $\ell$-bit message with an expected communication complexity of $O(n\ell + n^3 \log \alphabetsize )$  bits  and an expected  round complexity of $O(1)$ rounds,  under the optimal resilience  condition $n \geq 3t + 1$.

\subsection{Overview of $\OciorABA$}     

The proposed  $\OciorABA$ is described in Algorithm~\ref{algm:OciorABA}, together with  Algorithms~\ref{algm:APVA}-\ref{algm:ABBBA}. Fig.~\ref{fig:OciorABA} presents a block diagram of the proposed $\OciorABA$ protocol. Additionally, Fig.~\ref{fig:APVA} illustrates  a block diagram of the  $\APVA$ protocol, which  is utilized  in $\OciorABA$, while  Fig.~\ref{fig:ACIDstar} shows a block diagram of the   $\ACIDd$ protocol, which is employed in $\APVA$.   The main steps of $\OciorABA$ are outlined below.  
\begin{itemize}
\item    Erasure code  encoding: Each honest Node~$\thisnodeindex$ first encodes its initial message into $n$ symbols, $y_1^{(\thisnodeindex)}, y_2^{(\thisnodeindex)}, \cdots, y_{n}^{(\thisnodeindex)}$,  using an erasure code. 
\item    $\RBC$:  Then Node~$\thisnodeindex$ passes   $y_{\thisnodeindex}^{(\thisnodeindex)}$ into $\RBC[\ltuple \IDABA, \thisnodeindex \rtuple]$ as an input. 
    $\RBC[\ltuple \IDABA, \thisnodeindex \rtuple]$ represents  a   reliable broadcast  instance where Node~$\thisnodeindex$ is the leader.  $\RBC[\ltuple \IDABA, \thisnodeindex \rtuple]$ calls  the $\OciorRBC$ protocol \cite{ChenOciorCOOL:24}. 
\item    $\APVA$:   Upon  delivery of $y_j^{(j)}$  from   $\RBC[\ltuple \IDABA, j \rtuple]$, Node~$\thisnodeindex$ sets   $\APBVAinput_i[j] \gets 1$ if $y_j^{(j)} = y_j^{(\thisnodeindex)}$; otherwise,  Node~$\thisnodeindex$ sets  $\APBVAinput_i[j] \gets 0$.  Then, Node~$\thisnodeindex$ passes   $\APBVAinput_i[j]$ into $\APVA[\IDABA]$ as an input.  
\item     Erasure code  decoding:   Upon  delivery of output $\APBVAoutput$ from $\APVA[\IDABA]$, let $\ABAOneSet  \subseteq [1:n]$ denote the indices of all elements in $\APBVAoutput$    that  are  equal to  $1$, i.e., $\ABAOneSet=\{j:  \APBVAoutput[j] = 1, j \in[1:n]\}$. 
\begin{itemize}
\item  If $|\ABAOneSet| < t+1$, then each honest node outputs   $\defaultvalue$ and terminates. 
\item  If $|\ABAOneSet| \geq  t+1$,    let $\ABAOneSetK  \subseteq \ABAOneSet$ denote the first $t+1$ smallest  values in   $\ABAOneSet$.
In this case, each honest node waits for the  delivery of $y_j^{(j)}$  from   $\RBC[\ltuple \IDABA, j \rtuple]$, $\forall j\in \ABAOneSetK$, and then decodes  
	   $\hat{\Me} \gets \ECDec(n, t+1, \{y_j^{(j)}\}_{j\in \ABAOneSetK})$      using erasure code decoding. Each honest node     then 
outputs  $\hat{\Me}$ and terminates. 
\end{itemize} 
\end{itemize}

\subsection{Analysis of $\OciorABA$}

The main results of  $\OciorABA$ are summarized in the following theorems.

 \begin{theorem}  [Termination]  \label{thm:OciorABAfastterminate}
Given $n\geq 3t+1$,  each  honest node  eventually  outputs a  message and terminate   in $\OciorABA$.  
\end{theorem}
\begin{proof}
From Lemma~\ref{lm:OciorABAsameAPVAout}, in $\OciorABA$, every  honest node eventually outputs the same vector $\APBVAoutput$ from $\APVA[\IDABA]$, for some $\APBVAoutput$. 
Based on this conclusion,  the condition in Line~\ref{line:OciorABAAPVAtrigger} of Algorithm~\ref{algm:OciorABA} is eventually triggered at each honest node.  At this point, if the condition in  Line~\ref{line:OciorABAabaLESSt1} of Algorithm~\ref{algm:OciorABA} is satisfied, i.e.,   $|\ABAOneSet| < t+1$,  then each honest node eventually outputs  $\defaultvalue$ and terminates,   where $\ABAOneSet:=\{j:  \APBVAoutput[j] = 1, j \in[1:n]\}$.  If  $|\ABAOneSet| \geq t+1$,  each honest node eventually runs the steps in Lines~\ref{line:OciorABAabaGEQt1a}-\ref{line:OciorABAabaGEQt1b} of Algorithm~\ref{algm:OciorABA}  and then outputs a message and terminates.  It is worth noting that, based on Lemma~\ref{lm:OciorABAfastpropertyABBARBC},   $t+1$ instances of   $\{\RBC[\ltuple \IDABA, j \rtuple]\}_{j\in \ABAOneSetK}$ eventually deliver the same  outputs at  each honest node,  
where $\ABAOneSetK  \subseteq \ABAOneSet$ denotes the first $t+1$ smallest  values in   $\ABAOneSet$.   
 \end{proof}

 \begin{theorem}  [Validity]  \label{thm:OciorABAfastvalidity}
Given $n\geq 3t+1$, if  all  honest nodes input the same  value $\wv$, then each honest node eventually outputs $\wv$ in $\OciorABA$.   
\end{theorem}
\begin{proof} 
 From Lemma~\ref{lm:OciorABAsameAPVAout},     in $\OciorABA$  every  honest node eventually outputs the same vector $\APBVAoutput$ from $\APVA[\IDABA]$, for some $\APBVAoutput$.   
 By combining the results of Lemma~\ref{lm:OciorABAsameAPVAout},  Lemma~\ref{lm:APVAAgreement} (Consistency property of $\APVA$),  and Lemma~\ref{lm:APVAValidity} (Validity property of $\APVA$),  in $\OciorABA$ every honest node eventually  outputs the same vector $\APBVAoutput$ from $\APVA[\IDABA]$, such that the number of  non-missing  elements of  the output $\APBVAoutput$ is greater than or equal to  $n-t$, i.e.,   $|\NonmissingElementSet(\APBVAoutput)|\geq n-t$,  where $\NonmissingElementSet(\APBVAoutput):=\{j:  \APBVAoutput[j] \neq \missing, j \in[1:n]\}$.     
 The above conclusion, together with the conclusion from Lemma~\ref{lm:OciorABAfastvalidityHone}, reveals that  if  all  honest nodes input the same  value $\wv$, then  in $\OciorABA$ every honest node eventually  outputs the same vector $\APBVAoutput$ from $\APVA[\IDABA]$ such that $|\NonmissingElementSet(\APBVAoutput)|\geq n-t$ and that   $\APBVAoutput[j] =1$, $\forall j\in \NonmissingElementSet(\APBVAoutput)\setminus \Fc$.    Since $|\NonmissingElementSet(\APBVAoutput)\setminus \Fc| \geq n-2t \geq t+1$,  it further implies that if  all  honest nodes input the same  value $\wv$, then every honest node eventually runs the steps in Lines~\ref{line:OciorABAabaGEQt1a}-\ref{line:OciorABAabaGEQt1b} of Algorithm~\ref{algm:OciorABA} and outputs the decoded message.

Let  $\ABAOneSet:=\{j:  \APBVAoutput[j] = 1, j \in[1:n]\}$ and  let $\ABAOneSetK  \subseteq \ABAOneSet$ denote the first $t+1$ smallest  values in   $\ABAOneSet$.  
From the Validity property of $\APVA$ (see  Lemma~\ref{lm:APVAValidity}),  for any  $j\in \ABAOneSetK$,   at least one honest Node~$i$,   for $i\in [1:n]\setminus \Fc$, must have input $\APBVAinput_i[j] =\APBVAoutput[j] =1$, which also  implies that $y_j^{(j)} = y_j^{(i)}$   (see Line~\ref{line:abainputabbaconditionequal}  of Algorithm~\ref{algm:OciorABA}).  Thus, if  all  honest nodes input the same  value $\wv$, then it holds true that $y_j^{(j)}=\ECEnc_j(\networksizen,   t +1 , \wv), \forall j \in \ABAOneSetK$, where $\ECEnc_j(\networksizen,   t +1 , \wv)$ denotes the  $j$th  symbol encoded from $\wv$.  
Therefore,  if  all  honest nodes input the same  value $\wv$, then every honest node eventually runs the steps in Lines~\ref{line:OciorABAabaGEQt1a}-\ref{line:OciorABAabaGEQt1b} and outputs the same decoded message $\wv$ (see  Line~\ref{line:OciorABAabaECdec} of Algorithm~\ref{algm:OciorABA}), where $\wv=\ECDec(n, t+1, \{\ECEnc_j(\networksizen,   t +1 , \wv)\}_{j\in \ABAOneSetK})$. 
\end{proof}

\begin{theorem}  [Consistency]   \label{thm:OciorABAfastconsistency}
Given $n\geq 3t+1$,  if any honest node outputs a value $\wv$, then every honest node eventually outputs $\wv$   in $\OciorABA$.
\end{theorem}
\begin{proof} 

From Theorem~\ref{thm:OciorABAfastterminate}, every  honest node  eventually  outputs a  message and terminate   in $\OciorABA$.  
Due to the Consistency property of $\APVA[\IDABA]$, every honest node eventually outputs the  same vector $\APBVAoutput$  from $\APVA[\IDABA]$    (see Lemma~\ref{lm:OciorABAsameAPVAout}).
Due to the Consistency and Totality properties of $\RBC$,  every honest node eventually outputs  the same  coded symbol $y_j^{(j)}$ from   $\RBC[\ltuple \IDABA, j \rtuple]$, $\forall j\in \ABAOneSetK$ (see Lemma~\ref{lm:OciorABAfastpropertyABBARBC}),    where $\ABAOneSetK$ denotes the first $t+1$ smallest  values in   $\ABAOneSet:=\{j:  \APBVAoutput[j] = 1, j \in[1:n]\}$. 
If the condition in  Line~\ref{line:OciorABAabaLESSt1} of Algorithm~\ref{algm:OciorABA} is satisfied, i.e.,   $|\ABAOneSet| < t+1$,  then each honest node eventually outputs  $\defaultvalue$ and terminates. 
If  $|\ABAOneSet| \geq t+1$,  each honest node eventually runs the steps in Lines~\ref{line:OciorABAabaGEQt1a}-\ref{line:OciorABAabaGEQt1b} of Algorithm~\ref{algm:OciorABA}  and  outputs the same decoded message $\wv$ (see  Line~\ref{line:OciorABAabaECdec} of Algorithm~\ref{algm:OciorABA}), where $\wv=\ECDec(n, t+1, \{y_j^{(j)}\}_{j\in \ABAOneSetK})$.  
\end{proof}

\begin{theorem}  [Communication and Round Complexities]  \label{thm:OciorABAComplexites}
The expected communication complexity of  $\OciorABA$ is $O(n\ell + n^3 \log \alphabetsize )$ bits,  while the expected round complexity of  $\OciorABA$ is $O(1)$ rounds, given $n\geq 3t+1$.       Here $\alphabetsize$ denotes the alphabet size of  the error correction code used in the proposed  protocol.  
\end{theorem}
\begin{proof}
In $\OciorABA$,  each coded symbol  $y_j^{(\thisnodeindex)}$  (see Line~\ref{line:OciorABAabaECncoding} of Algorithm~\ref{algm:OciorABA}) carries $\symbolsize= \lceil \max\{\ell/k, \log \alphabetsize^{\star}\} \rceil$ bits, for $j, \thisnodeindex\in [1:n]$,  where $\alphabetsize^{\star}$ is the alphabet size of the erasure code used in the protocol, and $k=t+1$. 
By calling the $\OciorRBC$ protocol in \cite{ChenOciorCOOL:24}, the communication cost of each  $\RBC$ instance  for a coded symbol,  i.e., $\RBC[\ltuple \IDABA, j \rtuple]$ for $j\in[1:n]$,   is  $O(n\symbolsize + n t \log \alphabetsize)$ bits, where $\alphabetsize$ is the alphabet size of the error correction code used in the $\OciorRBC$ protocol.
 Therefore, the  communication cost of $n$  $\RBC$ instances is   $O(n^2\symbolsize + n^2 t \log \alphabetsize) = O(n\ell   + n^2 \log \alphabetsize^{\star} + n^2 t \log \alphabetsize) $ bits. 
 
 The  expected communication cost of the $\APVA$ protocol is   $O(n^3 + n^2 t \log \alphabetsize)$ bits. In  the $\APVA$   protocol,  the  communication cost of $n$  $\RBC$ instances,  	 i.e.,   $\RBC[\ltuple \IDABANew, j \rtuple]$  for  $\ConfirmRecord_{j}, \forall j \in[1:n]$ (see Line~\ref{line:APVARBCinputthisnode} of Algorithm~\ref{algm:APVA}), is $O(n^3 + n^2 t \log \alphabetsize)$ bits, where   $|\ConfirmRecord_{j}|=O(n), \forall j \in[1:n]$.  The other communication cost in $\APVA$ is expected $O(n^3)$ bits. 
  
 Therefore, the expected total    communication complexity of  $\OciorABA$ is $O(n\ell + n^3 \log \alphabetsize )$ bits,  considering $t=O(n)$.  
The expected round complexity of $\OciorABA$ is $O(1)$ rounds.  
\end{proof}

  \begin{lemma}     \label{lm:OciorABAsameAPVAout}
In $\OciorABA$, and given $n\geq 3t+1$, every  honest node eventually outputs the same vector $\APBVAoutput$ from $\APVA[\IDABA]$, for some $\APBVAoutput$. 
\end{lemma}
\begin{proof}
At first we will argue that, in $\OciorABA$,  at least one honest node  eventually outputs a   vector $\APBVAoutput$ from $\APVA[\IDABA]$.   
Let us assume that no honest node has output a vector  from $\APVA[\IDABA]$ yet.  Under this assumption,   at least $n-t$ instances of  $\RBC$  eventually deliver outputs at each honest node, resulting from the Validity property of $\RBC$.   
Then, all honest nodes eventually  input non-missing values in their input vectors for at least $n-t$ common positions   in the protocol $\APVA[\IDABA]$ (see Lines~\ref{line:OciorABAabainputabbacondition}-\ref{line:OciorABAabainputabba} of Algorithm~\ref{algm:OciorABA}), which implies that eventually every node   outputs a vector from $\APVA[\IDABA]$ (see Lemma~\ref{lm:APVAtermination}). 
From the Consistency property of $\APVA$ (Lemma~\ref{lm:APVAAgreement}),  if  any honest node outputs a vector $\APBVAoutput$  from $\APVA[\IDABA]$, then every honest node eventually outputs the same $\APBVAoutput$  from $\APVA[\IDABA]$,  for some $\APBVAoutput$. 
\end{proof}

 \begin{lemma}     \label{lm:OciorABAfastvalidityHone}
Given $n\geq 3t+1$, if  all  honest nodes input the same  value $\wv$ into $\OciorABA$, and an honest node outputs $\APBVAoutput$ from $\APVA[\IDABA]$, then   $\APBVAoutput[j] \neq  0$, $\forall j\in [1:n]\setminus \Fc$. 
\end{lemma}
\begin{proof}
If  all  honest nodes input the same  value $\wv$ into $\OciorABA$, then no honest node will input  $\APBVAinput_i[j] = 0$ into $\APVA[\IDABA]$, for any $i, j\in [1:n]\setminus \Fc$  (see Line~\ref{line:abainputabbaconditionequal} of Algorithm~\ref{algm:OciorABA}).  From this conclusion, in this case, if an honest node outputs $\APBVAoutput$ from $\APVA[\IDABA]$, then    $\APBVAoutput[j] \neq  0$, $\forall j\in [1:n]\setminus \Fc$, due to the Validity property of $\APVA$ (See Lemma~\ref{lm:APVAValidity}).   
\end{proof}

 \begin{lemma}  [Termination Property of $\APVA$]  \label{lm:APVAtermination}
In $\APVA$, if all honest nodes have input non-missing values in their input vectors for at least $n-t$  positions in common, then eventually every node will output a vector and terminate.
\end{lemma}
\begin{proof}
At first we argue that,  in $\APVA$, if all honest nodes have input non-missing values in their input vectors for at least $n-t$  positions in common, then at least one  honest node eventually outputs   values $[ \ReadyRecord^{(1)}, \ReadyRecord^{(0)},  \FinishRecord^{(1)},  \FinishRecord^{(0)}, \RBCReadyindicator,  \RBCFinishindicator]$	  from    $\ACIDd[ \IDMVBA ]$.   
Let us assume that  each  honest Node~$i$, for $i\in [1:n]\setminus \Fc$,  has input a non-missing value $\APBVAinput_i[j]\neq \missing$  into   $\APVA[\IDABA]$,  $\forall j\in \CommonSet$, for  some set $\CommonSet \subseteq [1:n]$ such that  $|\CommonSet| \geq n-t$. 
Let us also assume that  no honest node has output values  from $\ACIDd[ \IDMVBA ]$ yet. 
In this case,  each honest Node~$i$, for $i\in [1:n]\setminus \Fc$, eventually sets the values of $\ReadyRecord^{(\binaryvalue_j)}[j] \gets 1$, $\FinishRecord^{(\binaryvalue_j)}[j] \gets 1$ and $\ConfirmRecord_{\thisnodeindex}[j] \gets \binaryvalue_j$  in Lines~\ref{line:ACIDdReady},  \ref{line:ACIDdFinish},  and \ref{line:ACIDdConfirm}  of Algorithm~\ref{algm:ACIDd},  for some $\binaryvalue_j\in \{0,1\}$, and for any $ j\in \CommonSet$. Then,   the condition   $\ConfirmCount= n-t$  in  Line~\ref{line:ACIDdcntCond}   of Algorithm~\ref{algm:ACIDd} is eventually satisfied at each honest node. Each honest Node~$i$ then eventually sets $\RBCReadyindicator   [j] \gets 1$, $\RBCFinishindicator   [j] \gets 1$,  and sends $(\ELECTION, \IDMVBA)$  in   Lines~\ref{line:ACIDdReadyRBC},  \ref{line:ACIDdFinishRBC},  and \ref{line:ACIDdElectionSend}  of Algorithm~\ref{algm:ACIDd},  for any $ j\in \CommonSet$.  In this case, at least  honest node eventually outputs   values	  from    $\ACIDd[ \IDMVBA ]$.  
If one honest node returns values from $\ACIDd[ \IDMVBA ]$, then  every honest node eventually returns values from $\ACIDd[ \IDMVBA ]$  (see Lines~\ref{line:ACIDdConfirmEleBegin} -\ref{line:ACIDdConfirmEleEnd}   of Algorithm~\ref{algm:ACIDd}).

Furthermore,  when one honest node outputs   values	  from    $\ACIDd[ \IDMVBA ]$,  there exists a set $\Ic^{\star}$ such that the following conditions hold: 1) $\Ic^{\star}\subseteq [1:n]\setminus \Fc$; 2)  $|\Ic^{\star}| \geq n-2t$; and  3) for any $i\in \Ic^{\star}$,  $\Node_i$ has completed the dispersal $\ACD[\ltuple \IDMVBA, i \rtuple]$, i.e.,  $\Node_i$ has sent $(\ELECTION, \IDMVBA)$ to  all nodes  (see Line~\ref{line:ACIDdElectionSend} of Algorithm~\ref{algm:ACIDd}).

Since every  honest node  eventually returns values from $\ACIDd[ \IDMVBA ]$, then  every honest node eventually activates the condition in  Line~\ref{line:APVAACIDoutput} of Algorithm~\ref{algm:APVA} and  runs the steps in Lines~\ref{line:APVAABAround}-\ref{line:APVAout}   of Algorithm~\ref{algm:APVA}.  
Let us assume that no honest node has output  a vector   from $\APVA$   before some round $\electionround$, and that $\Election[\ltuple \IDMVBA,  \electionround \rtuple] \to \electionoutput$  (see Line~\ref{line:APVAElection} of Algorithm~\ref{algm:APVA}), where $\electionoutput\in \Ic^{\star}$. We will argue that, if one honest node reaches Round $\electionround$,  it is guaranteed that every honest node eventually outputs a vector   from $\APVA$ and terminates at Round $\electionround$. It is worth noting that, if  one honest node reaches Round $\electionround$, then eventually every honest node will reach Round $\electionround$ before termination. This follows from the Consistency property of   the protocols $\Election[\ltuple \IDABA,  \electionround \rtuple]$,  $\ABBA[\ltuple \IDABANew,   \electionoutput \rtuple]$, and $\ABBA[\ltuple \IDABA,   \electionoutput \rtuple]$, as well as the  Consistency  and Totality properties of  the protocol $\RBC[\ltuple \IDABANew, \electionoutput \rtuple]$  (see Lines~\ref{line:APVAElection}, \ref{line:APVAABBA}, \ref{line:APVARBCout} and \ref{line:APVAABBA2}  of Algorithm~\ref{algm:APVA}). 

At Round $\electionround$, with $\Election[\ltuple \IDMVBA,  \electionround \rtuple] \to \electionoutput$ and $\electionoutput\in \Ic^{\star}$,  it holds true that $\Node_\electionoutput$ has sent $(\ELECTION, \IDMVBA)$ to  all nodes  (see Line~\ref{line:ACIDdElectionSend} of Algorithm~\ref{algm:ACIDd}).  This implies that at least $n-2t$ honest nodes have set $\RBCFinishindicator   [j] \gets 1$   (see Lines~\ref{line:ACIDdFinishRBC} and \ref{line:ACIDdElectionSendCond} of Algorithm~\ref{algm:ACIDd}). Thus, based on the Biased Validity property of  $\ABBBA$ (see Lemma~\ref{lm:OciorABBBA}), every honest node eventually outputs $1$ from $\ABBBA[\ltuple \IDABANew,   \electionoutput, 0\rtuple]$ (see Line~\ref{line:APVAABBBA} of Algorithm~\ref{algm:APVA}) and then eventually outputs $1$ from $\ABBA[\ltuple \IDABANew,  \electionoutput \rtuple]$ (see Line~\ref{line:APVAABBA} of Algorithm~\ref{algm:APVA}). 
Then, at this round every honest node eventually receives  the same vector $\ConfirmRecord_{\electionoutput}$	         from  $\RBC[\ltuple \IDABANew, \electionoutput \rtuple]$ (see Line~\ref{line:APVARBCout} of Algorithm~\ref{algm:APVA}).  Since $\electionoutput\in \Ic^{\star}$, it means that at least $n-2t$ honest nodes have set $ \FinishRecord^{(\ConfirmRecord_{\electionoutput}[j])}[j]      \gets 1$, $\forall  j \inset \NonmissingElementSet(\ConfirmRecord_{\electionoutput}) $    (see Lines~\ref{line:ACIDdFinish} and \ref{line:ACIDdConfirmCond} of Algorithm~\ref{algm:ACIDd}), where $\NonmissingElementSet(\ConfirmRecord_{\electionoutput}):=\{j:  \ConfirmRecord_{\electionoutput}[j] \neq \missing, j \in[1:n]\}$.    
Thus, based on the Biased Validity property of  $\ABBBA$, every honest node eventually outputs $1$ from $\ABBBA[\ltuple \IDMVBA,   \electionoutput, j\rtuple]$ (see Line~\ref{line:APVAABBBAinput} of Algorithm~\ref{algm:APVA}), $\forall  j \inset \NonmissingElementSet(\ConfirmRecord_{\electionoutput})$,  and then eventually outputs $1$ from $\ABBA[\ltuple \IDMVBA,  \electionoutput\rtuple]$ (see Line~\ref{line:APVAABBA2} of Algorithm~\ref{algm:APVA}).     
Since $\electionoutput\in \Ic^{\star}$, it is also guaranteed that $|\NonmissingElementSet(\ConfirmRecord_{\electionoutput})| \geq n-t$  (see Line~\ref{line:ACIDdcntCond} of Algorithm~\ref{algm:ACIDd}).  Therefore,  at Round $\electionround$,   every node eventually outputs the same vector $\ConfirmRecord_{\electionoutput}$ and terminates.  
 \end{proof}

 \begin{lemma}  [Consistency Property of $\APVA$]  \label{lm:APVAAgreement}
In $\APVA$,  if  any honest node outputs a vector $\wv$, then every honest node eventually outputs the same vector $\wv$, for some $\wv$.
\end{lemma}
\begin{proof}
 
In $\APVA$, if any two honest nodes output values at Rounds  $\electionround$ and $\electionround'$ (see Line~\ref{line:APVAABAround}   of Algorithm~\ref{algm:APVA}), respectively, then $\electionround=\electionround'$. This follows from  the Consistency property of   the protocols $\Election[\ltuple \IDABA,  \electionround \rtuple]$,  $\ABBA[\ltuple \IDABANew,   \electionoutput \rtuple]$, and $\ABBA[\ltuple \IDABA,   \electionoutput \rtuple]$, as well as the  Consistency  and Totality properties of  the protocol $\RBC[\ltuple \IDABANew, \electionoutput \rtuple]$  (see Lines~\ref{line:APVAElection}, \ref{line:APVAABBA}, \ref{line:APVARBCout} and \ref{line:APVAABBA2}  of Algorithm~\ref{algm:APVA}).  
Moreover,  at  the same round $\electionround$,  if any two honest nodes output $\wv'$ and $\wv''$, respectively, then  $\wv'=\wv''$, due to the Consistency property of the protocol   $\RBC[\ltuple \IDABANew, \electionoutput \rtuple]$  (see Lines~\ref{line:APVARBCout}  and \ref{line:APVAout}   of Algorithm~\ref{algm:APVA}).

On the other hand, if one  honest node outputs a vector from $\APVA$, then this node must have returned values $[ \ReadyRecord^{(1)}, \ReadyRecord^{(0)},  \FinishRecord^{(1)},  \FinishRecord^{(0)}, \RBCReadyindicator,  \RBCFinishindicator]$	  from    $\ACIDd[ \IDMVBA ]$   (see Lines~\ref{line:APVAACIDoutput}   and \ref{line:APVAout}   of Algorithm~\ref{algm:APVA}). 
Furthermore, if one honest node returns values from $\ACIDd[ \IDMVBA ]$, then  every honest node eventually returns values from $\ACIDd[ \IDMVBA ]$  (see Lines~\ref{line:ACIDdConfirmEleBegin} -\ref{line:ACIDdConfirmEleEnd}   of Algorithm~\ref{algm:ACIDd}). 
Therefore, if one  honest node outputs a vector $\wv$ from $\APVA$, then  every honest node eventually activates the condition in  Line~\ref{line:APVAACIDoutput} of Algorithm~\ref{algm:APVA} and  runs the steps in Lines~\ref{line:APVAABAround} -\ref{line:APVAout}   of Algorithm~\ref{algm:APVA}.  
Thus, if one  honest node outputs a vector $\wv$ from $\APVA$ at  some round $\electionround$, then  every honest node eventually outputs the same vector $\wv$ at the same   round $\electionround$, due to the  Consistency property of   the protocols $\Election[\ltuple \IDABA,  \electionround \rtuple]$,  $\ABBA[\ltuple \IDABANew,   \electionoutput \rtuple]$, and $\ABBA[\ltuple \IDABA,   \electionoutput \rtuple]$, as well as the  Consistency  and Totality properties of  the protocols $\RBC[\ltuple \IDABANew, \electionoutput \rtuple]$.  
 \end{proof}

  \begin{lemma}  [Validity Property of $\APVA$]  \label{lm:APVAValidity}
In $\APVA$, if an honest node outputs $\APBVAoutput$, then for any non-missing element of $\APBVAoutput$, i.e., $\APBVAoutput[j] \neq \missing$ for some $j\in [1:n]$,   at least one honest Node~$i$,   for $i\in [1:n]\setminus \Fc$, must have input $\APBVAinput_i[j] =\APBVAoutput[j]\neq \missing$.   Furthermore,   the number of  non-missing  elements of  the output $\APBVAoutput$ is greater than or equal to  $n-t$, i.e.,   $|\NonmissingElementSet(\APBVAoutput)|\geq n-t$,  where $\NonmissingElementSet(\APBVAoutput):=\{j:  \APBVAoutput[j] \neq \missing, j \in[1:n]\}$.     
\end{lemma}
\begin{proof}
In $\APVA$, if an honest node outputs a vector $\APBVAoutput=\ConfirmRecord_{\electionoutput}$  (see Line~\ref{line:APVAout} of Algorithm~\ref{algm:APVA}) at some round $\electionround$ with $\Election[\ltuple \IDMVBA,  \electionround \rtuple] \to \electionoutput$  (see Line~\ref{line:APVAElection} of Algorithm~\ref{algm:APVA}),  then it holds true that 
$\ABBA[\ltuple \IDMVBA,  \electionoutput\rtuple]$  outputs   
$\ABAoutput'  =1$      (see Lines~\ref{line:APVAABBA2} and \ref{line:APVAoutCondition} of Algorithm~\ref{algm:APVA}), which further implies that at least one honest node outputs $\ABBAoutput_{\electionoutput, j} =1$ from $\ABBBA[\ltuple \IDMVBA,   \electionoutput, j  \rtuple]$, $\forall j\in \NonmissingElementSet(\ConfirmRecord_{\electionoutput})$   (see Lines~\ref{line:APVAABBBAinput} and  \ref{line:APVAABBBAoutput} of Algorithm~\ref{algm:APVA}). 
The instance of $\ABBAoutput_{\electionoutput, j} =1$ also reveals that at least one honest node inputs    $\ReadyRecord^{(\ConfirmRecord_{\electionoutput}[j])}[j]=1$ or $\FinishRecord^{(\ConfirmRecord_{\electionoutput}[j])}[j]=1$ into $\ABBBA[\ltuple \IDMVBA,   \electionoutput, j \rtuple]$,  based on the Biased Integrity  property of $\ABBBA$ (see Lemma~\ref{lm:OciorABBBA}).  
  When one honest node inputs    $\ReadyRecord^{(\ConfirmRecord_{\electionoutput}[j])}[j]=1$ or $\FinishRecord^{(\ConfirmRecord_{\electionoutput}[j])}[j]=1$, it is guaranteed that at least one honest Node~$i$ must have input  $\APBVAinput_i[j] =\ConfirmRecord_{\electionoutput}[j] \in  \{1,0\} $   (see Lines~\ref{line:ACIDdVoteBegin}-\ref{line:ACIDdFinishSend} of Algorithm~\ref{algm:ACIDd}).

 Furthermore,  if an honest node outputs a vector $\ConfirmRecord_{\electionoutput}$, it is guaranteed that  the number of  non-missing  elements of  the output $\ConfirmRecord_{\electionoutput}$ is greater than or equal to  $n-t$, i.e.,   $|\NonmissingElementSet(\ConfirmRecord_{\electionoutput})|\geq n-t$   (see Line~\ref{line:APVANNMissingThreshold}  of Algorithm~\ref{algm:APVA}). 
 \end{proof}

 \begin{lemma}  \label{lm:OciorABAfastpropertyABBARBC}
In $\OciorABA$, if  $\APVA[\IDABA]$  outputs $\APBVAoutput$  at an honest node and $|\ABAOneSet| \geq t+1$, with $\ABAOneSet:=\{j:  \APBVAoutput[j] = 1, j \in[1:n]\}$,   then $\RBC[\ltuple \IDABA, j \rtuple]$   eventually delivers the same output  $y_j^{(j)}$ at each honest node, for some $y_j^{(j)}$,   $\forall j\in \ABAOneSetK$, where $\ABAOneSetK$ denotes the first $t+1$ smallest  values in   $\ABAOneSet$. 		
\end{lemma}
\begin{proof}
In $\OciorABA$, if  $\APVA[\IDABA]$  outputs $\APBVAoutput$  at an honest node, then for any $j\in \ABAOneSet$ with $\ABAOneSet:=\{j:  \APBVAoutput[j] = 1, j \in[1:n]\}$,   at least one honest Node~$i$,   for $i\in [1:n]\setminus \Fc$, must have input $\APBVAinput_i[j] =\APBVAoutput[j]=1$, based on the Validity property of $\APVA$ (see Lemma~\ref{lm:APVAValidity}).
The instance of $\APBVAinput_i[j] =1$ reveals that Node~$i$ has received  a coded symbol $y_j^{(j)}$ from $\RBC[\ltuple \IDABA, j \rtuple]$ (see Lines~\ref{line:OciorABAabainputabbacondition} and \ref{line:abainputabbaconditionequal} of Algorithm~\ref{algm:OciorABA}). 
Due to the Totality and Consistency properties of $\RBC$, if one  honest node outputs a symbol $y_j^{(j)}$, then every honest node  eventually outputs the same symbol $y_j^{(j)}$.  
 \end{proof}

 \begin{lemma} [Properties of $\ABBBA$]    \label{lm:OciorABBBA}
The $\ABBBA$  protocol (Algorithm~2 of \cite{ChenOciorMVBA:24}, restated in Algorithm~\ref{algm:ABBBA} here)  satisfies the properties of Conditional Termination,  Biased Validity,  and Biased Integrity.
\end{lemma}
\begin{proof}
In the $\ABBBA$  protocol (see  Algorithm~\ref{algm:ABBBA}),  each  honest node inputs a pair of binary numbers $(\abbainputA, \abbainputB)$, for some $\abbainputA, \abbainputB \in \{0,1\}$. The  honest nodes seek to reach an agreement on a common value $\abbainput \in \{0,1\}$.
 This $\ABBBA$ protocol   satisfies the following properties:
\begin{itemize}
\item   \emph{Conditional Termination:} Under an input condition---i.e.,  if one honest node inputs its second number as $\abbainputB =1$ then at least $t+1$ honest nodes  input  their first numbers as $\abbainputA =1$---every honest node eventually outputs a value and terminates.  Under the above input condition, every honest node eventually outputs a value and terminates  in Lines~\ref{line:ABBBAoutput1},  \ref{line:ABBBAoneoutput} or \ref{line:ABBBAzero} of   Algorithm~\ref{algm:ABBBA}. 
\item   \emph{Biased Validity:} If at least $t+1$ honest nodes input the second number as $\abbainputB=1$, then any honest node that terminates outputs $1$.       
 If at least $\networkfaultsizet+1$ honest nodes input $\abbainputB=1$, then no honest node will output $0$. This is because  at most $\networksizen-(\networkfaultsizet+1)$ nodes input $\abbainputB=0$ in this case, indicating that the condition in Line~\ref{line:zerocondition} of  Algorithm~\ref{algm:ABBBA} cannot be satisfied.  
\item   \emph{Biased Integrity:} If any honest node outputs $1$, then at least one honest node inputs $\abbainputA=1$ or $\abbainputB=1$.        If one honest node outputs $1$, then the condition in Line~\ref{line:ABBBAoneoutputCond}  of  Algorithm~\ref{algm:ABBBA} is satisfied, which indicates that  at least one honest node has an input of  $\abbainputA =1$ or  $\abbainputB =1$.            
\end{itemize}  
 \end{proof}

%\bibliographystyle{IEEEtran}
%\bibliography{IEEEabrv,final_refs}


\end{document}